\documentclass[journal,doublecolumn,10pt]{IEEEtran}

\usepackage{epsfig,latexsym}
\usepackage{float}
\usepackage{indentfirst}
\usepackage{amsmath}
\usepackage{amssymb}
\usepackage{times}
\usepackage{subfigure}
\usepackage{psfrag}
\usepackage{cite}
\usepackage{lastpage}
\usepackage{fancyhdr}
\usepackage{color}

\newtheorem{Proposition}{Proposition}

\newtheorem{Corollary}{Corollary}
\newtheorem{theorem}{$\mathbf{Theorem}$}

\newtheorem{proposition}[Proposition]{Proposition}
\newtheorem{corollary}[Corollary]{$\mathbf{Corollary}$}

\begin{document}%
\title{ {\huge Wireless Information and Power Transfer in Cooperative Networks with Spatially Random Relays}}

\author{ Zhiguo Ding, \IEEEmembership{Member, IEEE},
Ioannis Krikidis, \IEEEmembership{Senior Member, IEEE}, Bayan Sharif, \IEEEmembership{Senior Member, IEEE}, and  H. Vincent Poor, \IEEEmembership{Fellow, IEEE}\thanks{
Part of this work was presented at the 2014 International Conference on Communications, Sydney,  June 2014. The work of Z. Ding was supported by a Marie Curie   International
Fellowship within the 7th European Community Framework Programme and
the UK EPSRC under grant number EP/I037423/1. The work of H. V. Poor was supported by the U. S. Air Force Office of Scientific Research
under MURI Grant FA9550-09-1-0643. I. Krikidis was supported by the Research Promotion Foundation, Cyprus under the project KOYLTOYRA/BP-NE/0613/04 ``Full-Duplex Radio: Modeling, Analysis and Design (FD-RD)''.

Z. Ding and H. V. Poor are with the Department of
Electrical Engineering, Princeton University, Princeton, NJ 08544,
USA.  I. Krikidis is with the Department of Electrical Engineering, Cyprus University, Cyprus. Z. Ding is also with the School of
Electrical, Electronic, and Computer Engineering, Newcastle
University, NE1 7RU, UK. B. Sharif is with College of Engineering, Khalifa University, Abu Dhabi, UAE.   }} \maketitle\vspace{-4em}
\begin{abstract}
 In this paper, the application of wireless information and power transfer to cooperative networks is investigated, where the relays in the network are   randomly located and based on the decode-forward strategy. For the scenario with one source-destination pair, three different strategies for using the available relays are studied, and their impact on the outage probability and diversity gain is characterized by applying stochastic geometry. By using the assumptions that  the path loss exponent is two and that the relay-destination distances are much larger than the source-relay distances,  closed form analytical results can be developed to demonstrate that the use of energy harvesting relays can achieve the same     diversity gain as the case with conventional self-powered relays.  For the scenario with multiple sources, the relays can be viewed as a type of scarce resource, where the sources compete with each other to get help from the relays. Such a competition is modeled as a coalition formation game, and two distributed game theoretic algorithms are developed based on different payoff functions.   Simulation results are provided to confirm the accuracy of the developed analytical results and facilitate a better performance comparison.
\end{abstract}

\section{Introduction}
Energy harvesting technologies have been recognized as a promising cost-effective solution to maximize the lifetime of wireless energy constrained  networks by eliminating   the cost for hard-wiring or replacing batteries of mobile nodes.  Conventional energy harvesting techniques scavenge energy from the environment, and therefore they are not applicable to the scenario in which wireless nodes do not have  any access to external   energy sources. This difficulty motivates  the recently developed concept of simultaneous  wireless  information and power transfer (SWIPT) \cite{varshney08, Grover10, Zhouzhang13, Nasirzhou,Ruizhangbroadcast13}.

The concept of SWIPT was first proposed in \cite{varshney08} and \cite{Grover10}, where it is assumed that  the receiver circuit can perform two functions, energy harvesting and information decoding, at the same time. Following these pioneering works, more practical receiver architectures have been developed by assuming that the receiver has two circuits to perform energy harvesting and information decoding separately \cite{Zhouzhang13,Nasirzhou}. Particularly the receiver either switches on  two circuits at different time, a strategy called time switching, or splits its observations into the two streams which are directed to two circuits at the same time, a strategy called  power splitting. The work in \cite{Zhouzhang13} considers a simple single-input single-output scenario, and an extension to multi-input multi-output broadcasting scenarios is considered in \cite{Ruizhangbroadcast13}.

SWIPT has been demonstrated as  a general energy harvesting technique and applied to various types of wireless communication networks. For example, in \cite{Ruicogswipt}, the application of SWIPT to cognitive radio networks is considered, where users from secondary networks   perform   energy harvesting from the primary transmitters and deliver information to their own destinations opportunistically. The use of SWIPT in OFDM networks has also received a lot of attention due to the success of WiMAX and   3GPP-Long Term Evolution (LTE) \cite{Huangl13}. In \cite{Liuzhangmiso} the combination of SWIPT with secure communications has also been considered, where an optimal beamforming and power allocation solution has been proposed to avoid the source information being intercepted by the energy harvesting eavesdroppers.

In this paper, we consider the application of SWIPT to wireless cooperative networks, where the relay transmissions are powered by the energy harvested from the relay observations. The contribution of this paper is two-fold. {\it Firstly} we focus on   cooperative networks with one source-destination pair and multiple {  energy harvesting} relays, and the impact of SWIPT on the reception reliability is studied by taking the spatial randomness of the relay locations into consideration, unlike   existing works in  \cite{Nasirzhou} and \cite{Dingpoor133} which treat the distances as constants. Stochastic geometry is used to characterize the density function for the wireless channels of the randomly deployed relays, where the developed analytical results are shown to match the simulations. In addition, three different strategies to use the available relays are studied, and we demonstrate that a more sophisticated relay selection strategy can ensure better reception reliability, albeit  with a price of more system overhead to realize the required channel state information (CSI) assumption. By using the assumptions that  the path loss exponent is two and the relay-destination distances are much larger than the source-relay distances, closed form analytical results can be developed to demonstrate that the use of energy harvesting relays can achieve the same     diversity gain as the case with conventional self-powered relays. However, the provided asymptotic studies show that the outage probability with energy harvesting relays is worse than that with conventional relays. For example, when a randomly chosen relay is used, it can be shown that the outage probability decays at a rate of $\frac{\log SNR}{SNR^2}$, instead of $\frac{1}{SNR^2}$ as in conventional cooperative networks, where SNR denotes the signal to noise ratio.

{\it Secondly} we consider a more challenging cooperative scenario where multiple sources communicate with one common destination via multiple energy harvesting relays. In such a scenario, the relays can be viewed as a type of scarce resource, where the sources compete with each other to get help from the relays. Such a competition can be modeled as a coalition formation game, and two distributed game theoretic algorithms are developed based on different payoff functions. In addition, analytical results are provided to demonstrate that a user-fairness approach should consider not only the SNR gain that a relay can contribute  to a coalition, but also how significant this gain is in contrast to the overall SNR of the coalition. Therefore we can avoid a situation with unbalanced relay allocation, i.e. some coalitions are crowded but some  coalitions do not get any help from the relays.  Both analytical and numerical results are provided to demonstrate the outage performance and convergence of the proposed coalition formation algorithms.

This paper is organized as follows. In Section \ref{section one source} the energy harvesting cooperative scenario with one source node is consider, and three different strategies for using relays are investigated. In Section \ref{section multipe source} the cooperative scenario with multiple source nodes is studied, and a game theoretic approach for coalition formation is proposed. Numerical results are shown in
Section~\ref{section simulation} for performance evaluation and comparison.  Finally, concluding remarks are
given in Section~\ref{conclusion}. The mathematical
proofs are collected in the appendix.

\begin{figure}[!htp]
\begin{center} \subfigure[Cooperative networks with one source]{\includegraphics[width=0.4\textwidth]{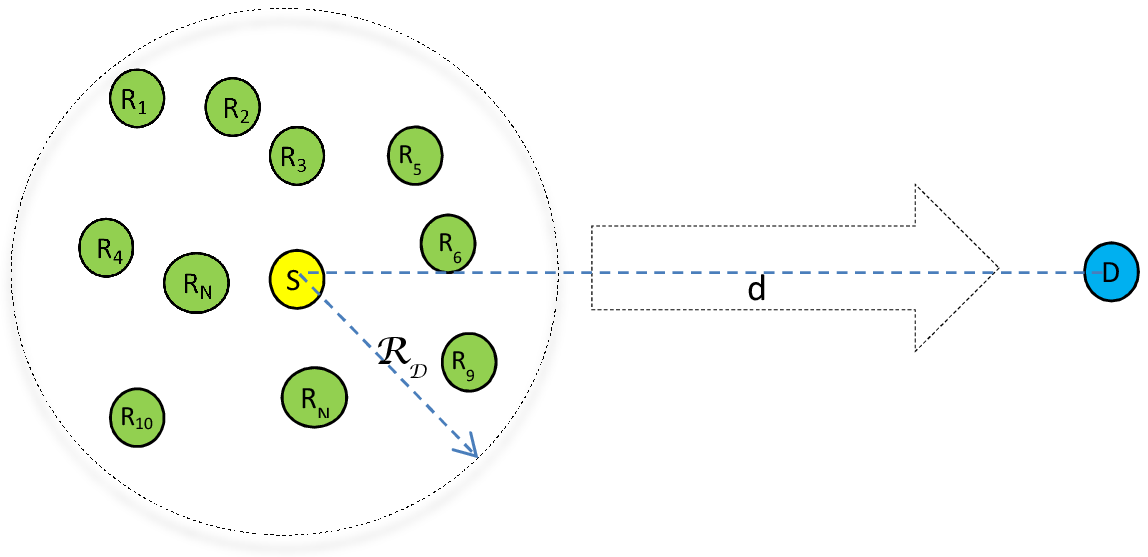}}
\subfigure[Cooperative networks with multiple sources ]{\label{fig set comparison
b2}\includegraphics[width=0.4\textwidth]{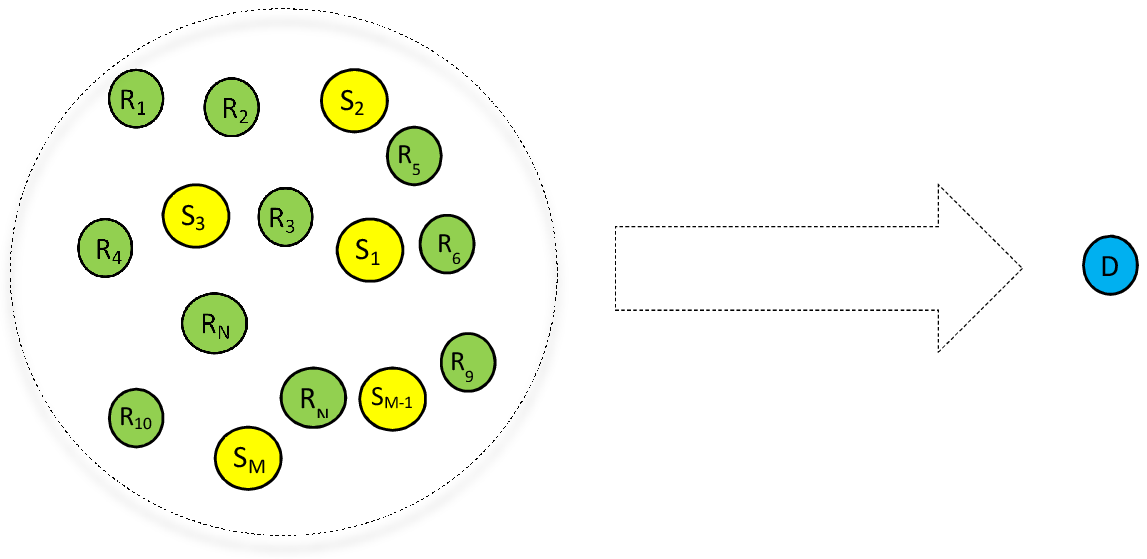}}
\end{center}\vspace{-1em}
 \caption{System diagrams for the scenarios considered in Section \ref{section one source} and Section \ref{section multipe source}. }\label{fig water}\vspace{-1em}
\end{figure}

\section{Energy Harvesting  Cooperative  Networks with One Source}\label{section one source}
Consider a cooperative scenario with one source-destination pair and multiple  {\it randomly deployed energy harvesting} relays, where the scenario with multiple sources will be considered in the next section.   In particular  consider a disc, denoted by $\mathcal{D}$, where the source is located at the origin of the disc and the radius of the disc is $R_\mathcal{D}$.  The location of the relays is modeled as a homogeneous Poisson point process with the intensity $\lambda_{\phi}$. Therefore the number of the relays in $\mathcal{D}$, denoted by $N$, is Poisson distributed, i.e. $\mathcal{P}(N=k)=\frac{\mu_\mathcal{D}^k}{k!}e^{-\mu_{\mathcal{D}}}$, where  $\mu_{\mathcal{D}}$ denotes the mean measure, i.e. $\mu_{\mathcal{D}}=\pi R_{\mathcal{D}}^2\lambda_{\phi}$. The distance between the source and the destination is denoted by $d$.

The decode-and-forward strategy is used at the relays, and the energy harvesting cooperative transmission consists of two phases. During the first time slot, the source broadcasts its message, and  all the relays and the destination listen to the source transmission. The energy harvesting relay will first try to direct the observation flow to the detection circuit, following from the power splitting approach in \cite{Zhouzhang13} and \cite{Nasirzhou}. If the detection is successful and there is any energy left, the remaining  signal flow will be directed to the energy harvesting circuit, and the harvested energy will be used to power the relay transmission. The observation split  to     the detection circuit is  given by
\begin{eqnarray}\label{signal model}
y_{ri} = \sqrt{(1-\theta_i)}\frac{h_i\sqrt{P}s}{\sqrt{1+d_i^\alpha}}+n_{r_i},
\end{eqnarray}
where $\theta_i$ is the power splitting coefficient,  $h_i$ models  frequency flat quasi static Rayleigh  fading, $d_i$ denotes the distance between the source and the $i$-th relay, $\alpha$ denotes the path loss exponent, $P$ is the transmission power, $s$ is the source message with the normalized  power and $n_{r_i}$ is  additive noise.  Note that  $\theta_i$ is used  to decide how much observation flow will be directed  to the energy harvesting circuit. For example, the choice of $\theta_i=0$ means that all observations will flow to the detection circuit, and $\theta_i=1$ means that the energy harvesting circuit receives  all of the observation flow.  The data rate supported by the source-relay channel is $R_i = \frac{1}{2}\log\left(1+ {(1-\theta_i)}\frac{|h_i|^2{P}}{ {1+d_i^\alpha}} \right)$.   Note that in \eqref{signal model}   the bounded path-loss model is used  to ensure that the  path loss is always larger than one for any distance \cite{Haenggi}, i.e. $1+d_i^\alpha>1$, even if $d<1$, whereas the simplified channel models used in \cite{Wangpoor11} and \cite{Shaodanma112} are valid only if the transceiver distance is larger than one.    To ensure successful detection at the relay given a targeted data rate $R$, i.e. $R_i= R$, the power splitting coefficient is set as follows:
\begin{equation}
\theta_i\triangleq  \max\left\{0,1-\frac{(1+d_i^\alpha)(2^{2R}-1)}{P|h_i|^2}\right\}.
\end{equation}
Note that the choice of $\theta_i$ in the above equation is based on the     strategy that a relay first tries to achieve information detection and then performs energy harvesting if there is any energy left. It is important to note  that different receiver strategies could result in different choices of $\theta_i$ as well as different values of the achieved outage performance. In addition, in this paper we do not consider how to use a relay that cannot decode the source information, but the use of such relays can potentially  yield more opportunities, particularly in two types of situations. The first is   the case with multiple pairs of sources and destinations. A relay that  cannot detect Source A's information can harvest energy from this source's signals, and use it to power the relay transmission to Source B. The second is when the the relay can store the energy harvested from the current time slot and use it for   future time slots. The consideration of different detection and energy harvesting strategies is beyond
the scope of this paper.

When the channel condition is poor,  i.e. $\theta=0$, no energy can be harvested from the observation since all the received signals will be directed to the detection circuit. When $\theta>0$,   the energy harvested at the $i$-th relay is given by
\begin{eqnarray}
E_{ri} =\frac{T\eta}{2} \left(\frac{|h_i|^2}{ {1+d_i^\alpha}} P - \tau \right),
\end{eqnarray}
where $\eta$ is the energy harvesting efficiency coefficient, $T$ is the time period for one time slot and  $\tau= 2^{2R}-1$.  It is assumed that the two phases of cooperative transmissions have the same time period.
So at the $i$-th relay,  the transmission power available for the second-stage relay transmission is
\begin{eqnarray}\label{power1}
P_{ri} =\frac{E_{ri}}{\frac{T}{2}}=\eta \left(\frac{|h_i|^2}{ {1+d_i^\alpha}} P - \tau \right).
\end{eqnarray}
Where there are multiple relays, i.e. $N>1$, it is of interest to study how to utilize these relays. In particular  we will study the performance of three   strategies with different CSI assumptions, as shown in the following three subsections.

\subsection{Random relay selection}
Prior to the transmissions, the source randomly selects a relay as its helper, a strategy that does not require any CSI.   Without loss of generality, consider that the $i$-th relay is selected to help the source.  This relay will use the harvested energy to power the relaying transmission, if it can decode the source message correctly. Therefore during the second time slot, the destination receives
\begin{eqnarray}
y_{D} = \frac{g_i}{\sqrt{1+c_i^\alpha}}\sqrt{P_{ri}}s+n_{D},
\end{eqnarray}
where $c_i$ is the distance between the $i$-th relay and the destination, and $g_i$ is the multi-path fading channel coefficient. After combining the observation from the first time slot, the receive SNR at the destination is given by
\begin{eqnarray}
SNR_i = \frac{|h_d|^2}{1+d^\alpha}P+\eta\frac{|g_i|^2}{ {1+c_i^\alpha}} \left(\frac{|h_i|^2}{ {1+d_i^\alpha}} P - \tau \right),
\end{eqnarray}
conditioned on a successful detection at the $i$-th relay, where $h_d$ is the multi-path fading channel. To simplify notation, we let $x_0\triangleq \frac{|h_d|^2}{1+d^\alpha}$, $x_i\triangleq \frac{|h_i|^2}{1+d_i^\alpha}$ and $y_i\triangleq \frac{|g_i|^2}{1+c_i^\alpha}$.

In this paper, the outage probability and diversity gain will be used for performance evaluation, as explained in the following. Provided  that the optimal channel coding scheme with  infinite coding length is used, the bit error probability can be closely bounded by the outage probability \cite{Zhengl03}. On the other hand, the diversity gain is an important metric for the robustness of transmissions in the high SNR regime.
  The outage probability can be
 Therefore the outage probability given the use of the $i$-th relay is
\begin{align}\nonumber
\mathcal{P}_i &\triangleq \mathcal{P}(N=0, P x_0<\tau) +\mathcal{P}\left( SNR_i < \tau, Px_i>\tau,N\geq 1\right)\\ \nonumber &+\mathcal{P}\left(P x_0<\tau,  Px_i<\tau,N\geq 1\right) \\ \nonumber &= \mathcal{P}(N=0,  x_0<\epsilon) + \underset{Q_1}{\underbrace{\mathcal{P}\left(x_0+ \eta y_i \left( x_i - \epsilon \right) < \epsilon, x_i>\epsilon\right)}}\\   &\times\mathcal{P}(N\geq 1) +\mathcal{P}(x_0<\epsilon, x_i<\epsilon,N\geq 1).\label{eqpi}
\end{align}
where $\epsilon=\frac{\tau}{\rho}$. The first probability in \eqref{eqpi} is for the event that there is no relay deployed in $\mathcal{D}$; the second and third ones are for the events that there is at least one relay in $\mathcal{D}$.  Note that the probability $Q_1$ is conditioned on $N\geq 1$, but such a condition can be omitted since it has no impact on the calculation of $Q_1$, as shown in the appendix.  Particulary the second probability is for the event that the $i$-th relay can detect the source message correctly but the overall SNR at the destination cannot support the targeted data rate; and the third one is for the event that neither the $i$-th relay nor the destination can  detect  the source message.  The following theorem provides an exact expression and a high-SNR approximation for the outage probability.
\begin{theorem}\label{thorem1}
The outage probability achieved by an energy harvesting cooperative protocol with a randomly chosen relay is given by \eqref{theorem},
\begin{figure*}
\begin{eqnarray}\nonumber
\mathcal{P} &=& \left(1-e^{-(1+d^\alpha)\epsilon}\right) e^{-\pi R_{\mathcal{D}}^2\lambda_{\phi}}  + \frac{\left(1- e^{-\pi R_{\mathcal{D}}^2\lambda_{\phi}} \right)}{\pi R_{\mathcal{D}}^2}
\int_{0}^{\epsilon}\int^{R_{\mathcal{D}}}_{0}\int^{2\pi}_{0}e^{-(1+r^\alpha)\epsilon}\left(1-2 q(r, \theta) \mathrm{K}_1\left(2q(r, \theta)\right) \right)
\\ \label{theorem} && \times rdrd\theta f_{x_0}(x_0) dx_0 +\frac{2}{ R_{\mathcal{D}}^2} \left(1- e^{-\pi R_{\mathcal{D}}^2\lambda_{\phi}} \right)\int_0^{R_{\mathcal{D}}}\left(1-e^{-(1+d^\alpha)\epsilon}\right)\left(1-e^{-(1+r^\alpha)\epsilon}\right)rdr,
\end{eqnarray}
\end{figure*}
where $f_{x_0}=\frac{1}{1+d^\alpha}e^{-(1+d^\alpha)x_0}$ and
$q(r, \theta)\triangleq \sqrt{\frac{(1+\left(r^2+d^2-2 rd \cos(\theta) \right)^\frac{\alpha}{2})(1+r^\alpha)(\epsilon-x_0)}{\eta}}$. For the special case of $\alpha=2$ and $R_{\mathcal{D}}<<d$, the outage probability can be approximated at high SNR as in \eqref{theorem approximation},
\begin{figure*}
\begin{eqnarray}\label{theorem approximation}
\mathcal{P} &\approx& (1+d^2)\epsilon e^{-\pi R_{\mathcal{D}}^2\lambda_{\phi}} - \frac{\eta a_1}{4  R_{\mathcal{D}}^2}
\left[R_{\mathcal{D}}^2(R_{\mathcal{D}}^2+2)\left(\frac{(1+d^2)^2\epsilon^2}{\eta^2} \ln \frac{(1+d^2)\epsilon}{\eta}-\frac{(1+d^2)^2\epsilon^2}{\eta^2}\right)+\frac{(1+d^2)^2\epsilon^2}{\eta^2} \right.\\ \nonumber &&\left. \times\left((1+R^2_{\mathcal{D}})^2\ln(1+R^2_{\mathcal{D}})     +4e_1 R_{\mathcal{D}}^2(R_{\mathcal{D}}^2+2)\right)\right]\left(1-e^{-\pi R_{\mathcal{D}}^2\lambda_{\phi}} \right)  +\frac{1}{2}(R_{\mathcal{D}}^2+2)(1+d^2)\epsilon^2\left(1-e^{-\pi R_{\mathcal{D}}^2\lambda_{\phi}} \right),
\end{eqnarray}
\end{figure*}
where $a_1=(1-(1+d^2)\epsilon)$, $e_1=-\frac{1}{4}\left(\psi(1)+\psi(2)\right)$ and $\psi(\cdot)$ denotes the psi function.
\end{theorem}
\begin{proof}
See the appendix.
\end{proof}
Note that the   exact expression of the outage probability shown in \eqref{theorem} is applicable for  any choices of $\alpha$ and distances. However, such a general expression is very complicated since it contains multiple integrals. Therefore the use of such an involved expression is not helpful for developing insights   about the fundamental limits  of energy harvesting relaying, which motivates  the studies for the special case with $\alpha=2$ and $R_{\mathcal{D}}<<d$,  i.e. the radius of $\mathcal{D}$ is much smaller than the source-destination distance. Numerical results demonstrate that the approximated analytical results in Theorem \ref{thorem1} are accurate when $d>5 R_{\mathcal{D}}$, as can be seen  in Section \ref{section simulation}. Theorem \ref{thorem1} can be used to study the diversity gain achieved by the energy harvesting cooperative scheme, as shown in the following corollary.
\begin{corollary}\label{corollary1}
For the special case with $\alpha=2$,  $R_{\mathcal{D}}<<d$ and  $N\geq 1$, the diversity gain achieved by the energy harvesting cooperative protocol with a randomly chosen  relay is $2$.
\end{corollary}
\begin{proof}
When $N\geq 1$, the first factor in \eqref{theorem approximation} can be ignored, as explained in the proof for Theorem \ref{thorem1}. Therefore  by applying $\epsilon\rightarrow 0$, the corollary can be obtained in a straightforward manner.
\end{proof}
Recall that in a conventional cooperative network, the use of a randomly chosen relay will also yield a diversity gain of $2$ \cite{Laneman04}. Corollary \ref{corollary1} states that the use of energy harvesting relays will not decrease the diversity gain of cooperative protocols.
However, an important observation from \eqref{theorem approximation} is that the dominant factor in the expression for the outage probability at high SNR is $-\epsilon^2\ln  \epsilon$, or equivalently $\frac{\ln SNR}{SNR^2}$. Therefore, the use of energy harvesting relays will cause the outage probability to decrease  at a rate of $\frac{\ln SNR}{SNR^2}$, whereas a faster decaying rate of $\frac{1}{SNR^2}$ can be realized in a conventional cooperative network.

\subsection{Relay selection based on the second order statistics of the channels}
For many practical communication scenarios, it is realistic to obtain the second order statistics of   wireless channels. Such information is determined by the distance between the transceivers and changes  more slowly compared to small scale multi-path fading. In this section, we will focus on the impact of relay selection on the outage probability when the second order statistics of the channels are known. To make meaningful conclusions, we assume  $N\geq1$, $\alpha=2$, and $R_{\mathcal{D}}<<d$, the same conditions as in Theorem \ref{thorem1}. With these conditions,   intuition suggests  that the optimal   strategy of relay selection is to find the relay closest  to the source, which can be confirmed in the following proposition.
\begin{proposition}\label{proposition 1}
Assume $N\geq1$, $\alpha=2$ and $R_{\mathcal{D}}<<d$.
Selecting a relay that  is closest to the source   minimizes the outage probability at   high SNR.
\end{proposition}
\begin{proof}
See the appendix.
\end{proof}
As can be seen from Proposition \ref{proposition 1}, the criterion for relay selection is   based only on the source-relay distances. This is due to the  assumption $R_{\mathcal{D}}<<d$ which leads to the approximation that all relay-destination distances are the same.

By using  the density function  of the shortest source-relay distance, we can obtain the following theorem  about the diversity gain achieved by the proposed relay selection scheme.
\begin{theorem}\label{theorem2}
Assume $\alpha=2$,  $R_{\mathcal{D}}<<d$ and  $N\geq 1$. The diversity gain achieved by the relay selection scheme based on the second order statistics of the channel  is $2$.
\end{theorem}
\begin{proof}
See the appendix.
\end{proof}
The exact expression of the outage probability is difficult to obtain, since   the use of the relay closest to the source makes the   density function of the source-relay channel more complicated. Therefore, in Theorem \ref{theorem2} we can   obtain the diversity order only after applying high-SNR asymptotic analysis.  Surprisingly the knowledge of the second order statistics of the channels is not helpful for  improving  the diversity order. However, it is worth pointing out that the use of this CSI knowledge can improve the outage performance compared to the scheme with a randomly chosen relay, as can be seen in Section \ref{section simulation}.

\subsection{Distributed beamforming}
When   global CSI is available at the source and the relays, an optimal strategy for using the available relays is to apply distributed beamforming, analogous  to single-input multiple-output  scenarios in which maximal  ratio combining is optimal. Again we assume that there are $N\geq 1$ relays in $\mathcal{D}$, and $\mathcal{N}=\{1, \ldots, N\}$. 
Denote the  group of   qualified  relays that can decode the source message correctly by $\tilde{\mathcal{S}}$, and the group containing the remaining  relays  by $\tilde{\mathcal{S}}^c$, i.e. $\mathcal{N}=\tilde{\mathcal{S}}\cup \tilde{\mathcal{S}}^c$. 
Note that it is possible that one of the two sets is empty.  The transmission strategy for a qualified relay is as follows. The relay $i$,  $i\in \tilde{\mathcal{S}}$, will transmit $\frac{g_i^*P_{ri}}{\sqrt{\xi(1+c_i^{\alpha})}} s$, where $\xi =\sum_{i\in \tilde{\mathcal{S}}} \frac{|g_i|^2P_{ri}}{{ (1+c_i^{\alpha})}} $. The power normalization factor $\xi$ is to enure that the transmission power of the relay $i$ is always less than $P_{ri}$. With such beamforming, during the second time slot, the destination observes
\begin{eqnarray}
y_D = \left(\sum_{i\in \tilde{\mathcal{S}}} \frac{g_i}{\sqrt{1+c_i^{\alpha}}}\frac{g_i^* {P_{ri}}}{\sqrt{\xi(1+c_i^{\alpha})}}\right) s+n_D.
\end{eqnarray}
By applying MRC over two time slots,   the resulting  SNR at the destination will be
\begin{eqnarray}
SNR_{\tilde{\mathcal{S}}} = \frac{|h_d|^2}{1+d^\alpha}P+\sum_{i\in \tilde{\mathcal{S}}}\eta\frac{|g_i|^2}{ {1+c_i^\alpha}} \left(\frac{|h_i|^2}{ {1+d_i^\alpha}} P - \tau \right).
\end{eqnarray}
 Define $\Pi_n$ as  a set containing all possible partitions yielding distinct pairs of  $\tilde{\mathcal{S}}$ and $\tilde{\mathcal{S}}^c$.  The overall outage probability conditioned on $N\geq 1$ is given by
\begin{align}
\mathcal{P}_{N} &=  \underset{\Pi_n}{\sum}  \mathcal{P}\left( \frac{1}{2}\log (1+SNR_{\tilde{\mathcal{S}}})<R ,  x_{i}>\epsilon, i\in \tilde{\mathcal{S}}, \right. \\ \nonumber &\left. x_{j}<\epsilon, j\in \tilde{\mathcal{S}}^c\right).
\end{align}
In conventional cooperative networks, the SNR at the destination is independent of the source-relay channels, which is no longer the case in the   energy harvesting network considered here.
Define $z_i\triangleq\frac{|g_i|^2}{ {1+c_i^\alpha}} \left(\frac{|h_i|^2}{ {1+d_i^\alpha}} P - \tau \right)$, and the   outage probability of the considered energy harvesting network then  becomes
\begin{align}\nonumber
 \mathcal{P}_{N} &=  \underset{\Pi_n}{\sum} \underset{x_0<\epsilon}{ \mathcal{E}}\underset{Q_5}{\underbrace{\left\{\mathcal{P}\left( \sum^{n}_{i=1}z_{i}<\frac{\tau-x_0P}{\eta}, x_{i}>\epsilon, i\in \tilde{\mathcal{S}} \right)\right\}}}\\ \label{outage 2} &\times \mathcal{P}\left(x_{j}<\epsilon, j\in \tilde{\mathcal{S}}^c\right),
\end{align}
where $\mathcal{E}\{\cdot\}$ denotes the expectation operation.
As can be observed from the proofs for Theorems \ref{thorem1} and \ref{theorem2}, it is quite difficult to find an exact expression of the pdf for $z_i$. The outage probability in \eqref{outage 2} requires not only the pdf of $z_i$ but also the pdf of $ \sum^{n}_{i=1}z_{i}$.  Therefore finding an exact expression for this outage probability is difficult. In order to obtain the diversity order achieved by distributed beamforming, we can first develop   lower and upper bounds on the outage probability in \eqref{outage 2} and then show that the   bounds converge  in the high SNR regime. These two steps will result  in   the following theorem.
\begin{theorem}\label{theorem3}
Assume $\alpha=2$,  $R_{\mathcal{D}}<<d$ and  $N\geq 1$. In an energy harvesting network with randomly deployed relays, the use of distributed beamforming   achieves the maximum diversity gain $(N+1)$.
\end{theorem}
\begin{proof}
See the appendix.
\end{proof}
As can be observed from Theorem \ref{theorem3}, the full diversity gain can   still be achieved, even though the relaying transmissions are powered by the energy harvested from the relay observations.

\section{Energy Harvesting Cooperative Networks With Multiple Sources}\label{section multipe source}
In this section, we will  consider a more general  cooperative   scenario in which $M$ sources communicate with a common destination via $N$ {\it energy harvesting} relays.  Similar to the previous section, the cooperative transmission consists of two phases. During the first phase, the $M$ sources first broadcast their messages, denoted by $s_m$ for source $m$, via orthogonal channels. During the second phase,  the $N$ relays will form $M$ disjoint groups to help the $M$ sources  via the orthogonal channels.  Distributed beamforming will be carried out among the relays from the same group since it  can achieve the maximum diversity gain as described in the previous section. Therefore a relay in a group to help the $m$-th source, denoted by $\mathcal{S}_m$, will send the following message:
\begin{eqnarray}
s_{R_i,i\in \mathcal{S}_m} = \left\{\begin{array}{ll} \frac{g_i^*\sqrt{P_{mi}}}{|g_i|\sqrt{1+c_{i}^\alpha}} s_m, &if \quad \frac{|h_{mi}|^2}{ {1+d_{mi}^\alpha}} P \geq \tau \\ 0 ,&otherwise \end{array}\right.,
\end{eqnarray}
 where the transmission power of the relay is powered by the energy harvested from its incoming signal, i.e. $P_{mi} =\max\left\{0,\eta \left(\frac{|h_{mi}|^2}{ {1+d_{mi}^\alpha}} P - \tau \right)\right\}$ as shown in \eqref{power1},   $d_{mi}$ is the distance between the $m$-th source and the $i$-th relay and $h_{mi}$ is the corresponding multi-path fading channel coefficient.

For the considered  multi-source scenario, the relays can be viewed as a type of scarce resource, where the sources compete with each other to get help from the relays. Such a competition can be modeled as a coalition formation game.  Particularly, denote by $\mathcal{N}\triangleq \{1, \ldots, N\}$   the set of all relays and similarly by $\mathcal{S}$  the set of all source nodes. Let $\mathcal{S}_m$    be a coalition consisting of the $m$-th source and the relays that are willing to help this source. Therefore $\sum^{M}_{m=1}\bar{\bar{\mathcal{S}_m}}= (N+M)$ and $1\leq \bar{\bar{\mathcal{S}_m}}\leq (N+1)$, where $\bar{\bar{\mathcal{S}_m}}$ denotes the cardinality of $\mathcal{S}_m$. A network partition is defined as $\Pi=\{\mathcal{S}_1, \ldots, \mathcal{S}_M\}$

\subsection{A baseline approach without considering user fairness}
The receive SNR is an important parameter  since it determines the   data rate as well as the reception reliability. Given that the relays in $\mathcal{S}_m$ perform distributed beamforming, the SNR for the $m$-th source message at the destination is given by
\begin{eqnarray}
 SNR_{\mathcal{S}_m}&\triangleq& \frac{P|h_{dm}|^2}{1+d_{0m}^\alpha}+\sum_{i\in \mathcal{S}_m} \frac{P_{mi} |g_i|^2}{ {1+c_i^\alpha}},
\end{eqnarray}
where $d_{0m}$ denotes the distance between the $m$-th source and the destination and $h_{dm}$ denotes the corresponding small scale fading channel.
A straightforward approach to opportunistically use the relays is to define the  payoff function for the $i$-th relay to join   $\mathcal{S}_m$    as follows:
\begin{eqnarray}\label{payoff1}
\phi_i(\mathcal{S}_m) &=& SNR_{\mathcal{S}_m} -  SNR_{\mathcal{S}_m\diagup i} - c(\mathcal{S}_m)\\ \nonumber &=&  \frac{P_{mi} |g_i|^2}{ {1+c_i^\alpha}} - c(\mathcal{S}_m),
\end{eqnarray}
where $SNR_{\mathcal{S}_m\diagup i}$ denotes the SNR achieved by removing the $i$-th relay from the coalition $\mathcal{S}_m$.
The cost, $c(\mathcal{S}_m)$,  is related to the size of the coalition, and in this paper we assume that the cost is proportional to the number of   relays performing distributed beamforming, i.e., $c(\mathcal{S}_m) = \kappa|\tilde{\mathcal{S}}_m|$, where
$\kappa$ is a coefficient to measure the cost to coordinate distributed beamforming and $\tilde{\mathcal{S}}_m\in \mathcal{S}_m$ contains all the relays in $\mathcal{S}_m$ that  can decode $s_m$ correctly.
Therefore  the value of each coalition is given by
\begin{eqnarray}\label{value1}
v(\mathcal{S}_m) = \sum_{i\in \mathcal{S}_m}\phi_i(\mathcal{S}_m).
\end{eqnarray}

The above definitions of the payoff and the coalition value leads to a solution that  does not consider the fairness among the users, as illustrated by the following example. Consider a scenario with $M=2$ sources and  $N=2$ relays, where there is no direct link between the sources and destination, i.e. $h_{dm}=0$. Assume that the first relay has a very good connection to the first source, but no connection to the second source, e.g., $\phi_1(\mathcal{S}_1)=1000 $ and $\phi_1(\mathcal{S}_1)=0 $, where the cost has been ignored.  The channel condition between the second relay and the first source is slightly better than that between the second relay and the second source, e.g., $\phi_2(\mathcal{S}_1)=50 $ and $\phi_2(\mathcal{S}_2)=49 $. The definitions in \eqref{payoff1} and \eqref{value1} imply  that the second relay will join  $\mathcal{S}_1$. However,   the contribution of the second relay in $\mathcal{S}_1$ is insignificant due to the fact that $\phi_1(\mathcal{S}_1)=1000 $, whereas including the second relay in $\mathcal{S}_2$ is important to achieve better fairness among the users.  This observation motivates the following approach which achieves a better tradeoff between the system performance and fairness.

\subsection{A user-fairness coalition formation  approach}
In order to take user fairness into consideration, consider the following alternative definition of the payoff function for the $i$-th relay joining $\mathcal{S}_m$:
\begin{eqnarray}\label{payoff2}
\phi_i(\mathcal{S}_m) =\frac{ SNR_{\mathcal{S}_m} -  SNR_{\mathcal{S}_m\diagup i} }{  SNR_{\mathcal{S}_m} }- c(\mathcal{S}_m).
\end{eqnarray}
And the value of the coalition is $
v(\mathcal{S}_m) = \sum_{i\in \mathcal{S}_m}\phi_i(\mathcal{S}_m)$.
Compared to the definition in \eqref{payoff1}, the one in \eqref{payoff2} can take  the user fairness into consideration, and encourage the relays to help the sources that  need the help more desperately. We consider  the same example as in the previous section, i.e., $\frac{P_{11} |g_1|^2}{ {1+c_1^\alpha}}=10000$,  $\frac{P_{21} |g_1|^2}{ {1+c_1^\alpha}}=0$,  $\frac{P_{12} |g_2|^2}{ {1+c_2^\alpha}}=50$,  $\frac{P_{22} |g_2|^2}{ {1+c_2^\alpha}}=49$,  $ c(\mathcal{S}_m)=0$ and $h_{dm}=0$. In such a case, the first relay always joins  $\mathcal{S}_1$. Based on the definition in \eqref{payoff2}, the payoffs for the second relay to join  $\mathcal{S}_1$ and $\mathcal{S}_2$ are $\phi_2(\mathcal{S}_1) =\frac{50}{1050}$ and $\phi_2(\mathcal{S}_2) =1$, respectively. Therefore the use of the new payoff function in \eqref{payoff2} can ensure that the second relay joins  $\mathcal{S}_2$ and help the second source which is in a critical situation of an outage. Additional  properties of  the proposed fairness approach will be discussed in the following subsection.  Note that there are other possible payoff functions other than the ones shown in \eqref{payoff1} and \eqref{payoff2}. The benefit of using these  two payoff functions is two-fold. One is that these payoff functions are based on  SNRs which are important parameters   directly related to various metrics for performance evaluation, such as the data rate and the outage probability. The other is that these payoff functions are linear functions of the SNRs, and hence can be easily used to analyze the various properties of the addressed games, as shown in the next sub-section.

\subsection{A distributed coalition formation   algorithm }
Based on the definitions  in \eqref{payoff1} and \eqref{payoff2},  one can observe that the payoff of a relay  depends only on the members of the coalition in which this relay is located. Therefore the proposed coalitional game can be modeled as a hedonic coalition formation game, in which the coalition formation process is accomplished by applying preference relations \cite{Bogomo,Saadhedic}. Particulary for any relay $i\in \mathcal{N}$, consider two coalitions $\mathcal{S}_m$ and $\mathcal{S}_n$, where $i\in \mathcal{S}_m$ and $i\in \mathcal{S}_n$. The preference relation, denoted by $\mathcal{S}_m\prec_i \mathcal{S}_n$ , implies that the relay prefers to join $\mathcal{S}_n$ instead of $\mathcal{S}_m$. In according with  \cite{Saadhedic}, we use the following two criteria to determine the preference relation
\begin{eqnarray}\label{criterion relation}
\mathcal{S}_m\prec_i \mathcal{S}_n \Leftrightarrow \left\{\begin{array}{l} \text{C1:}\quad \phi_i(\mathcal{S}_m)< \phi_i(\mathcal{S}_n)\\  \text{C2:}\quad v(\mathcal{S}_m)+v(\mathcal{S}_n\diagdown  \{i\}) \\ \quad \quad \quad <v(\mathcal{S}_m\diagdown \{i\})+v(\mathcal{S}_n)\end{array}\right.,
\end{eqnarray}
where $\mathcal{S}\diagdown \{i\}$ denotes a subset of $\mathcal{S}$ created by removing the node $i$.
The motivation to have the first criterion in \eqref{criterion relation} is that each relay tries to maximize its own individual benefit. And the second criterion in \eqref{criterion relation} is to impose a constraint that the overall network benefit will not be reduced if the relay moves from $\mathcal{S}_m$ to $\mathcal{S}_n$. These two criteria could be conflicting. For example, a relay wants to move from   $\mathcal{S}_m$ to $\mathcal{S}_n$ since its payoff will be increased, but such a move may be blocked since it will reduce the overall network benefit. However, the following proposition demonstrates that for some critical situations, C1 is a sufficient condition for  C2, i.e. a satisfaction  of C1 will lead to a satisfaction of C2.

\begin{proposition}\label{pro2}
Consider a scenario  in which   $\bar{\bar{\mathcal{S}_m}}>2$,  $\bar{\bar{\mathcal{S}_n}}=2$,  $i\in \mathcal{S}_m$, and $i\in \mathcal{S}_n$.   The cost for cooperation is ignored, i.e. $c(\mathcal{S}_m)=0$. If $\phi_i(\mathcal{S}_m)< \phi_i(\mathcal{S}_n)$, then $v(\mathcal{S}_m)+v(\mathcal{S}_n\diagdown  \{i\}) <v(\mathcal{S}_m\diagdown \{i\})+v(\mathcal{S}_n)$ also holds.
\end{proposition}
\begin{proof}
See the appendix.
\end{proof}

The situation described in Proposition \ref{pro2} is critical, as described in the following. The relay $i$ has a choice to help one of two sources, i.e. sources $m$ and $n$. The condition  $\bar{\bar{\mathcal{S}_n}}=2$ means that the $n$-th source does not get any help except from the relay $i$. The condition  $\phi_i(\mathcal{S}_m)< \phi_i(\mathcal{S}_n)$ means that the relay wants to move from $\mathcal{S}_m$ to $\mathcal{S}_n$, since its payoff will be increased. Such a move is critical to the $n$-th source since its SNR will be zero if  such a move is rejected. Proposition \ref{pro2} illustrates that such a move will be guaranteed.

By using the conditions in \eqref{criterion relation}, a distributed coalition formation algorithm can be described as follows. During the initialization phase, the relays are randomly assigned to the $M$ sources. During the iteration phase, each relay takes its turn to determine  whether to stay in the same coalition or join  a new coalition based on the criteria in \eqref{criterion relation}.
Compared to Eq. (12) in \cite{Saadhedic}, the conditions for the preference relation in \eqref{criterion relation} are weaker in the sense that assigning the relay $i$ to a coalition will reduce the payoffs of other players in the same coalition. But the conditions in \eqref{criterion relation} are sufficient to ensure the convergence of the proposed algorithm, which can be shown by using the criterion of  Nash-stability. Recall that a partition $\Pi=\{\mathcal{S}_1, \ldots, \mathcal{S}_M\}$ is Nash stable if $\forall i \in \mathcal{N}$ s.t. $i\in \mathcal{S}_m$, $(\mathcal{S}_m , \Pi) \succ_i (\mathcal{S}_k\cup\{i\}, \Pi')$ for any $\mathcal{S}_k$ with $\Pi'=(\Pi\diagdown \{\mathcal{S}_m,\mathcal{S}_k\}\cup \{\mathcal{S}_m\diagdown\{i\}\})$. The following proposition can demonstrate the convergence of the proposed scheme.
\begin{proposition}\label{converge}
Starting from any initial partition, the proposed coalition formation algorithm always converges to a final network partition, and this final partition is Nash stable.
\end{proposition}
\begin{proof}
The key step for the proof is to observe that  the overall network benefit is always non-decreasing after each iteration, because of the second criterion in \eqref{criterion relation}. With this observation and following similar steps as in \cite{Saadhedic}, the proposition can be proved.
\end{proof}

\section{Numerical Results}\label{section simulation}
To demonstrate the performance of the proposed energy harvesting cooperative schemes,  we present some numerical  studies and evaluate the developed analytical results   in two different scenarios.  The numerical results are obtained by carrying out Monte Carlo simulations, and the number of simulation runs for all figures except Fig. \ref{fig_selection} is $10^{5}$, whereas $10^{8}$ simulation runs are used for Fig. \ref{fig_selection}.
\subsection{Energy harvesting cooperative networks with one source}
In Fig. \ref{fig_selection}, the performance of the energy harvesting transmission scheme with a randomly chosen relay is depicted. Particularly the targeted data rate is $R=0.1$ bit per channel use (BPCU), the energy harvesting efficiency coefficient is $\eta=0.5$, the path loss exponent is $\alpha=2$  and the radius of $\mathcal{D}$ is $R_{\mathcal{D}}=1.5$m. The number of relays in $\mathcal{D}$ is Poisson distributed with the parameter $\lambda_\phi=1$. As can be observed from the figure, the analytical results developed in Theorem \ref{thorem1} are very close to the simulations, particularly at high SNR. In addition, the performance of a non-cooperative direct transmission scheme has also been shown. As can be seen from Fig. \ref{fig_selection}, the use of energy harvesting relays is helpful to improve the reception reliability at the destination. Particularly the slope of the outage probability curve for the cooperative scheme is larger than that of the non-cooperative one, which means that a larger diversity gain can be achieved by the cooperative scheme.

\begin{figure}[!htbp]\centering
    \epsfig{file=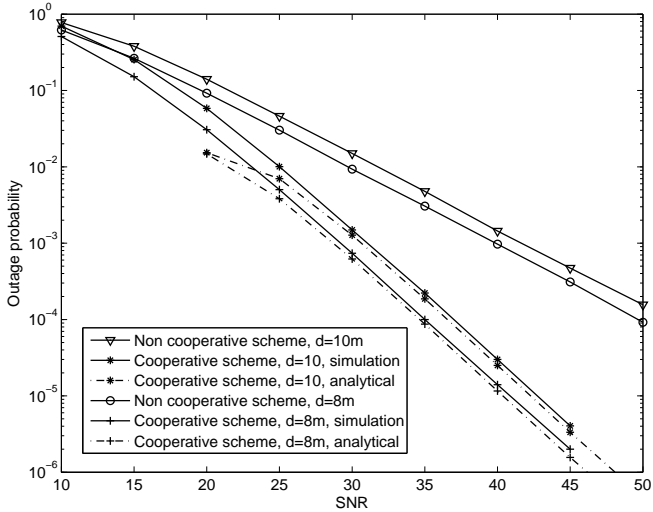, width=0.48\textwidth, clip=}
\caption{ Outage probability versus SNR with a randomly chosen relay. The targeted data rate is $R=0.1$ BPCU. The energy harvesting efficiency is $\eta=0.5$, the radius of $\mathcal{D}$ is $R_{\mathcal{D}}=1.5$m, and the node density is $\lambda_\phi=1$. Solid lines are for simulation results and dashed lines are for the analytical  results developed in Theorem 1.  }\label{fig_selection}\vspace{-1em}
\end{figure}

\begin{figure}[!htp]
\begin{center} \subfigure[ $\eta=0.5$]{\includegraphics[width=0.48\textwidth]{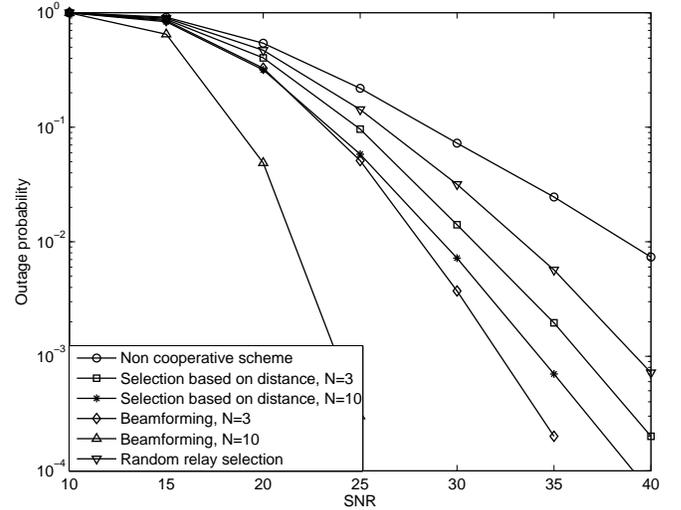}}
\subfigure[$\eta=1$ ]{\label{fig set comparison
b2}\includegraphics[width=0.48\textwidth]{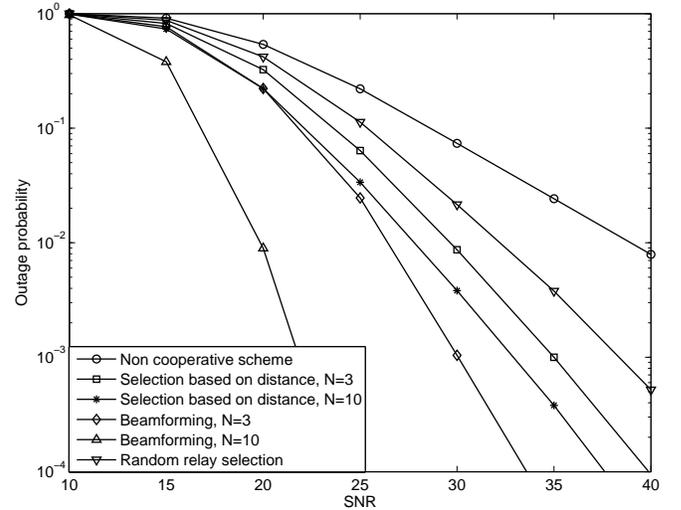}}
\end{center}\vspace{-1em}
 \caption{ Outage probability achieved by   different strategies of using the available relays.  $R=1$ BPCU, $R_{\mathcal{D}}=2.5$m and $d=5$m. All the curves are based on simulation results.  }\label{fig water}\vspace{-1em}
\end{figure}

In Fig. \ref{fig water}, the three different strategies   using the available relays are compared, when   the targeted data rate is $R=1$ BPCU, the source-destination distance is $5$m,  the radius of $\mathcal{D}$ is $R_{\mathcal{D}}=2.5$m and the other parameters are the same as in Fig. \ref{fig_selection}. Compared to the setup used in Fig. \ref{fig_selection}, the ratio between the source-destination distance and the radius of $\mathcal{D}$  is reduced in order to examine the performance of the proposed relaying schemes for the case other than  $R_{\mathcal{D}}<<d$.  The impact of   different numbers of   relays in $\mathcal{D}$ on the outage performance is shown in the two figures. Theorem \ref{theorem2} states that when only the second order statistics of the channels are known, the achievable diversity gain is $2$, no matter how many relays there are in $\mathcal{D}$.  In Fig. \ref{fig water}, it is clear that the slopes  of the outage curves for the distance based scheme with different numbers of relays are the same, which confirms Theorem \ref{theorem2}. However,  it is worth pointing out that the use of the second order statistics of the channels can still yield an outage performance gain compared to the case with a randomly chosen relay.  On the other hand, the use of distributed beamforming can ensure that the achievable diversity gain is proportional to $N$,  as can be observed from the figures. Such an observation   confirms the analytical results developed in Theorem \ref{theorem3}.

The developed analytical results for the diversity gains shown in  Theorem \ref{thorem1}, \ref{theorem2} and \ref{theorem3}   are based on the assumption $\alpha=2$, and in Fig. \ref{fig_alpha}, we use computer simulations to demonstrate the impact of the path loss exponent on the outage performance. As can be seen from the figure, by increasing the path loss exponent, the outage performance achieved by all the relaying protocols is degraded. However, an important observation is that the slope of the outage probability curves stays the same. Take the beamforming scheme as an example. By increasing the value of $\alpha$, the outage probability curve is shifted to the right, and its slope stays the same. Therefore, Fig. \ref{fig_alpha} has demonstrated that our developed diversity results are most likely valid even if $\alpha>2$, although we still do not have a formal proof of this.

\begin{figure}[!htbp]\centering
    \epsfig{file=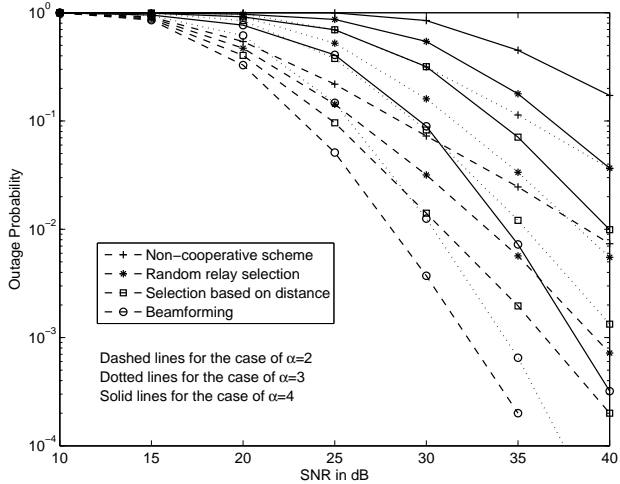, width=0.48\textwidth, clip=}
\caption{ The impact of the path loss exponent  $\alpha$ on the outage performance of the energy harvesting relaying protocols.  $R=1$ BPCU. $R_{\mathcal{D}}=2.5$m and $d=5$m.  All the curves are based on simulation results. }\label{fig_alpha}\vspace{-1em}
\end{figure}
\subsection{Energy harvesting cooperative networks with multiple sources}
In this subsection the energy harvesting scenario with multiple source nodes will be considered, and the performance of the proposed coalition formation algorithms will be evaluated. The relays are randomly located inside the disc $\mathcal{D}$.  The radius of $\mathcal{D}$ is $R_{\mathcal{D}}=5$m and the origin of $\mathcal{D}$ is located at $(R_{\mathcal{D}},R_{\mathcal{D}})$. We focus on the case with $4$ sources, which are located at $(\frac{1}{2}R_{\mathcal{D}},R_{\mathcal{D}})$, $(R_{\mathcal{D}},\frac{1}{2}R_{\mathcal{D}})$, $(\frac{3}{2}R_{\mathcal{D}},R_{\mathcal{D}})$ and $( R_{\mathcal{D}},\frac{3}{2}R_{\mathcal{D}})$, respectively. The common destination is located at $(R_{\mathcal{D}}+d,R_{\mathcal{D}})$ and $d=10$m.  In Fig. \ref{fig water1}, we evaluate the performance of two coalition formation schemes based on different payoff functions defined  in \eqref{payoff1} and \eqref{payoff2}, termed the Baseline Scheme and the Proposed Scheme, respectively. The energy harvesting efficiency is $\eta=1$ and the cost coefficient is $\kappa=0.001$.  Different choices of the number of relays and the targeted data rate are also shown in the figures. As can be observed from the figure, by increasing the targeted data rate, the outage performance achieved by both coalition formation algorithms is reduced.  When increasing the number of relays, i.e. increasing the value of $N$,  there are more relays to help source transmissions, and therefore one can expect that the outage performance achieved by both schemes should be improved, which can be confirmed by the two figures.

Another observation from the two figures is that the payoff function in \eqref{payoff2} can yield  better outage performance compared to the one based on \eqref{payoff1}. The reason for such a performance gain is that the payoff function in \eqref{payoff2}  can efficiently capture how significantly a relay   contributes to a coalition. Particularly the payoff function in \eqref{payoff2} evaluates the ratio between the SNR gain that a relay can bring to a coalition and the overall SNR achieved by the coalition. In the case that the overall SNR of one  coalition is already large enough, the payoff function in \eqref{payoff2} encourages the relay to find an alternative coalition in which this relay's help is more significant. As a result, such a payoff function can ensure  balanced relay allocation among the sources. As discussed in Proposition \ref{pro2}, the proposed coalition formation algorithm based on \eqref{payoff2} can avoid the unfair  situations in which all relays join   one coalition, while other coalitions  do not get any help from the relays.

\begin{figure}[!htp]
\begin{center} \subfigure[ $N=8$]{\includegraphics[width=0.48\textwidth]{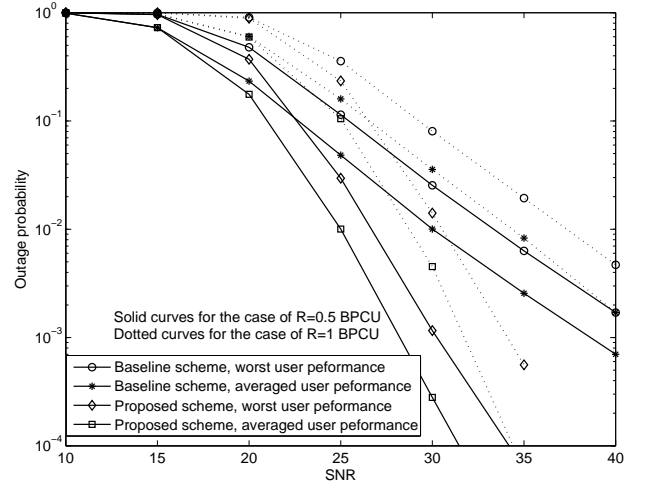}}
\subfigure[$N=4$ ]{\label{fig set comparison
b2}\includegraphics[width=0.48\textwidth]{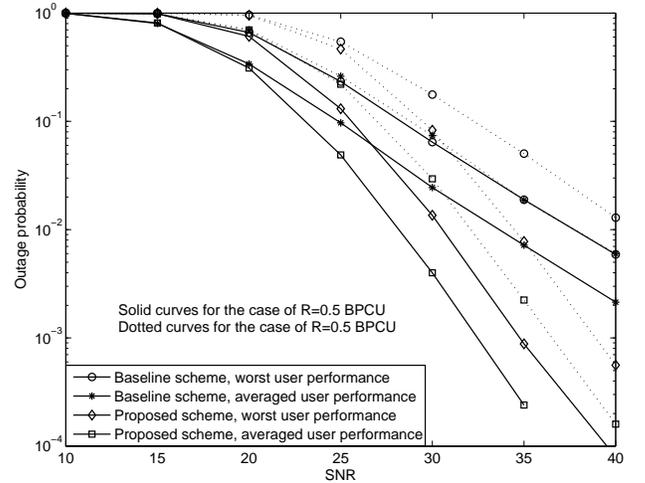}}
\end{center}\vspace{-1em}
 \caption{ Outage probability achieved by the proposed coalition formation algorithms. The cost for coordinating node cooperation is  $\kappa=0.001$, $R_{\mathcal{D}}=5$m and $d=10$m. All the curves are based on simulation results.  }\label{fig water1}\vspace{-1em}
\end{figure}

%

In Fig. \ref{fig_converge}, the convergence of the proposed coalition formation algorithm is studied, where the targeted data rate is set as $R=1$ BPCU, $N=8$ and the other parameters are   the same as in Fig. \ref{fig water1}. As can be seen from the figure, the proposed coalition formation algorithm can converge quickly, which is important to reduce the delay and computational complexity of the energy harvesting cooperative system. In Fig. \ref{fig_cost}, the impact of different choices of the cost coefficient $\kappa$ on the outage performance is shown. As can be observed from the figure, increasing the value of the cost coefficient will reduce the outage probability, which can be explained as follows. By increasing the value of $\kappa$, the coordination of the same coalition will result in more system overhead, which implies that each source prefers to reduce the coalition size.  As a result, each source will have fewer relays to help its transmission, which causes the degradation of the outage performance.  In Fig. \ref{fig_cost2}, the outage probability is shown as a function of $\kappa$, where the choices of $\kappa$ are $\begin{bmatrix} 0.001 &0.005 &0.01 &0.05 &0.1 &0.5 &1  \end{bmatrix}$. Again we observe that the outage probability will be increased by enlarging $\kappa$, similar to the observations from  Fig. \ref{fig_cost}. An interesting observation is that the outage probability is   sensitive to the choice of $\kappa$ when $\kappa$ is of the order of $0.001$, whereas the outage probability curves are quite flat for $\kappa$ of the order of $0.01$.
\begin{figure}[!htbp]\centering
    \epsfig{file=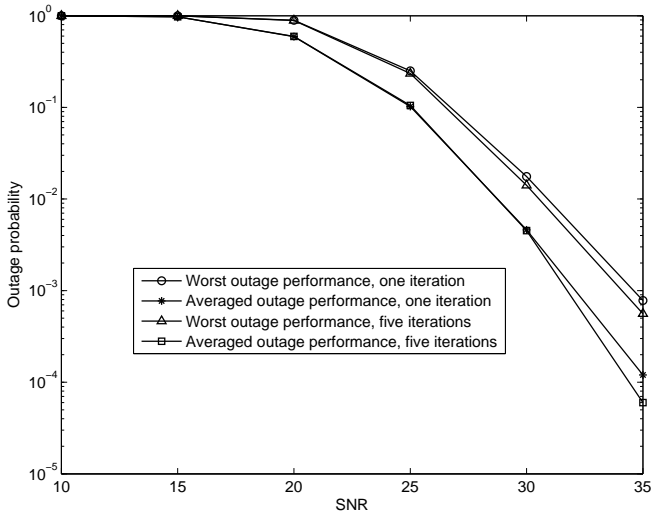, width=0.48\textwidth, clip=}
\caption{ The convergence of the proposed coalition formation algorithm. $R=1$ BPCU. $\kappa=0.001$ and $N=8$. All the curves are based on simulation results. }\label{fig_converge}\vspace{-1em}
\end{figure}

\begin{figure}[!htbp]\centering
    \epsfig{file=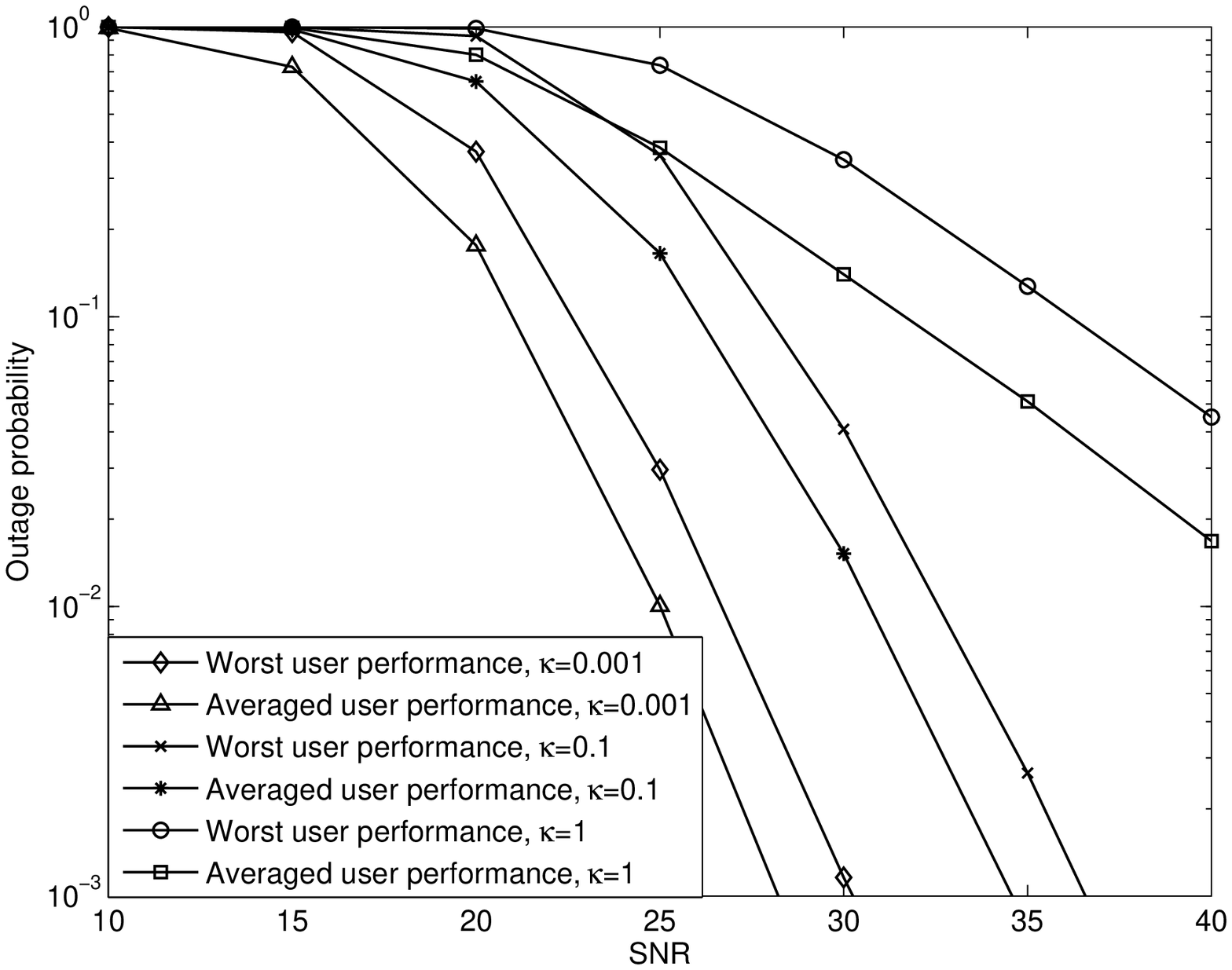, width=0.48\textwidth, clip=}
\caption{ The impact of the cost coefficient $\kappa$ on the performance achieved by the proposed coalition formation algorithm. $R=1$ BPCU and $N=8$.  All the curves are based on simulation results. }\label{fig_cost}\vspace{-1em}
\end{figure}

\begin{figure}[!htbp]\centering
    \epsfig{file=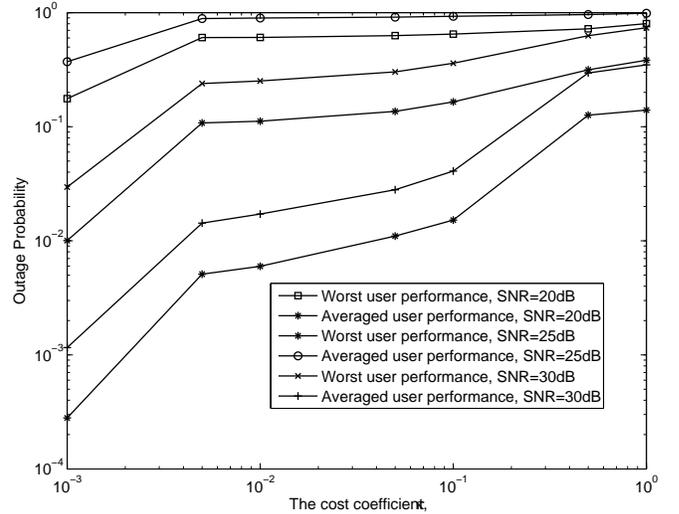, width=0.48\textwidth, clip=}
\caption{ The impact of the cost coefficient $\kappa$ on the performance achieved by  proposed coalition formation algorithm. $R=1$ BPCU and $N=8$.  All the curves are based on simulation results. }\label{fig_cost2}\vspace{-1em}
\end{figure}
%

\section{Conclusion}\label{conclusion}
 In this paper, we have considered the application of wireless information and power transfer to cooperative networks with spatially random relays. When there is a single  source-destination pair in the network, we have proposed three different strategies to use the available relays, and their impact on the outage probability and diversity gain has been characterized by applying stochastic geometry.   When there are multiple sources, relay allocation among the users has been modeled as a coalition formation game, and two distributed game theoretic algorithms have been  developed based on different payoff functions.    Most of the analytical results developed in this paper rely on the assumption that the relays are located far away from the destination, where a promising future direction is to characterize the outage performance for cooperative networks in a more general topology.  A relevant  topic of interest  is whether it is more   beneficial to use the relays deployed close to the source, compared to the ones  close to the destination. A useful observation is that  the source-relay channel     not only affects the reliability of the transmission from the source to the relay, but also plays an important role for the transmission from the relay to the destination, since the relay transmission power is determined by the source-relay channel condition. In addition, it is promising to study the use of amplify-and-forward (AF) strategies at the relays. The use of AF will significantly complicate the determination of $\theta_i$. Thus, it will be important   to find the optimal choice of $\theta_i$ in this case  and also to understand the corresponding outage performance.

 Note that this paper has considered  the strategy of simultaneous information and power transfer, and this strategy can  be more energy efficient compared to the decoupled strategy, in which some sources  send only information and the other sources are placed only to  deliver  energy. Take the cooperative network described in Section \ref{section one source} as an example, i.e. one source-destination pair with multiple relays. When the source broadcasts its message, some relays can first use   part of their observations for decoding, and then carry out energy harvesting by using the rest after successful decoding. As a result, the energy broadcast by the source can be fully used. On the other hand, the energy efficiency of the system will be reduced if we consider the decoupled strategy, since the potential energy remaining  in the relay observations cannot be used efficiently.

 \bibliographystyle{IEEEtran}
\bibliography{IEEEfull,trasfer}

\appendix
\textsl{Proof of Theorem \ref{thorem1} :}
Recall that the $N$ relays are deployed in $\mathcal{D}$ according to a homogeneous Poisson point process.  Therefore the $N$ relays can be modeled as a set of independent and identically distributed points in the disc $\mathcal{D}$, denoted by $W_i$. The point $W_i$ contains the location information about the  relay, from  which the source-relay and relay-destination distances can be calculated. And the probability density function (pdf) of each point $W_i$ is given by \cite{Kingmanbook,Wangpoor11}
\[
p_{W_i}(w_i) = \frac{\lambda_{\phi}}{\mu(\mathcal{D})} = \frac{1}{\pi R_{\mathcal{D}}^2}.
\]
By using this density function, the probability in \eqref{eqpi},  $Q_1$, can be re-written as
\begin{align}\label{q1}
Q_1 &=  \mathcal{P}\left(x_0+ \eta y_i \left( x_i - \epsilon \right) < \epsilon, x_i>\epsilon\right)
\\\nonumber &=  \mathcal{P}\left(  y_i   < \frac{\epsilon-x_0}{\eta(x_i - \epsilon)}, x_i>\epsilon\right) \\ \nonumber &= \int_{0}^{\epsilon}\underset{\mathcal{D}}{\int}\underset{Q_2}{\underbrace{\int^{\infty}_{\epsilon }\left(1-e^{-(1+c_i^\alpha)\frac{\epsilon-x_0}{\eta(t-\epsilon)} }\right)  (1+d_i^\alpha)e^{- (1+d_i^\alpha)t}dt}}\\ \nonumber &\times  p_{W_i}(w_i)dw_if_{x_0}(x_0)dx_0,
\end{align}
where   $f_{x_0}(x_0)$ is the pdf of $x_0$.  Since $x_0=\frac{|h_d|^2}{1+d^\alpha}$ and $h_d$ is assumed to be complex Gaussian distributed, we have $f_{x_0}(x_0)=(1+d^\alpha)e^{-(1+d^\alpha)x_0}$.  Note that the parameter $d_i$ is determined by the location of the relay, i.e. $W_i$, as shown later.

To obtain a closed-form expression, the integral $Q_2$ can be first re-written as follows:
\begin{eqnarray}
Q_2  &=& e^{-(1+d_i^\alpha)\epsilon}-\int^{\infty}_{0}e^{-\frac{(\epsilon-x_0)(1+c_i^\alpha)}{\eta z} }  \\ \nonumber && \times (1+d_i^\alpha)e^{- (1+d_i^\alpha)(z+\epsilon)}dz \\ \nonumber&=&
 e^{-(1+d_i^\alpha)\epsilon}-(1+d_i^\alpha) e^{- (1+d_i^\alpha)\epsilon}  \\ \nonumber && \times \int^{\infty}_{0}e^{-\frac{(\epsilon-x_0)(1+c_i^\alpha)}{\eta z} }e^{- (1+d_i^\alpha)z }dz.
\end{eqnarray}
Now applying [Eq.3,324] in \cite{GRADSHTEYN}, the integral can be expressed as follows:
\begin{eqnarray}
Q_2  &=&\nonumber
 e^{-(1+d_i^\alpha)\epsilon}\left(1-2 \left(\sqrt{\frac{(1+c_i^\alpha)(1+d_i^\alpha)(\epsilon-x_0)}{\eta}}\right)\right. \\   &&\left.\times \mathrm{K}_1\left(2\sqrt{\frac{(1+c_i^\alpha)(1+d_i^\alpha)(\epsilon-x_0)}{\eta}}\right) \right),
\end{eqnarray}
where $\mathrm{K}_n(\cdot)$ denotes the modified Bessel function of the second kind.  Given $d_i$, the source-relay distance, $d$, the source-destination distance and the angle $\angle R_iSD$, denoted by $\theta_i$, the relay-destination distance is given by
\[
c_i^2 = d_i^2+d^2-2 d_id \cos(\theta_i).
\]
By using the above relationship, the probability $Q_1$ can be written as follows:
\begin{align}
Q_1  &=\nonumber
\int_{0}^{\epsilon}\underset{\mathcal{D}}{\int}e^{-(1+d_i^\alpha)\epsilon}\left(1-2 q(d_i, \theta_i) \mathrm{K}_1\left(2q(d_i, \theta_i)\right) \right)\\ \nonumber &\times p_{W_i}(w_i)dw_if_{x_0}(x_0)dx_0.\nonumber
\end{align}
where $q(d_i, \theta_i)$ is defined as $q(d_i, \theta_i)\triangleq \sqrt{\frac{(1+\left(d_i^2+d^2-2 d_id \cos(\theta_i) \right)^\frac{\alpha}{2})(1+d_i^\alpha)(\epsilon-x_0)}{\eta}}$.
By using the homogenous Poisson point process assumption and converting to polar coordinates,   we have
\begin{align}
Q_1  &=\label{eq3} \frac{1}{\pi R_{\mathcal{D}}^2}
\int_{0}^{\epsilon}\int^{R_{\mathcal{D}}}_{0}\int^{2\pi}_{0}e^{-(1+r^\alpha)\epsilon}\left(1-2 q(r, \theta) \right.\\ \nonumber &\times \left.\mathrm{K}_1\left(2q(r, \theta)\right) \right)rdrd\theta f_{x_0}(x_0) dx_0.
\end{align}

On the other hand, the probability of the event that both source-relay and source-destination channels are   poor   is given by
\begin{align}\label{eq2}
&\mathcal{P}(x_0<\epsilon, x_i<\epsilon)\\\nonumber &=\underset{\mathcal{D}}{\int}\left(1-e^{-(1+d^\alpha)\epsilon}\right)\left(1-e^{-(1+d_i^\alpha)\epsilon}\right)p_{W_i}(w_i)dw_i\\ \nonumber
&=\frac{2}{ R_{\mathcal{D}}^2}\int_0^{R_{\mathcal{D}}}\left(1-e^{-(1+d^\alpha)\epsilon}\right)\left(1-e^{-(1+r^\alpha)\epsilon}\right)rdr,
\end{align}
from which the third probability in \eqref{eqpi} can be obtained.
Recall that from the Poisson distribution, we have
\[
P(N=0) = e^{-\pi R^2\lambda_{\phi}}.
\]
Combining the above equation with \eqref{eqpi}, \eqref{eq3} and \eqref{eq2}, the first part of the theorem is proved.

To obtain the high SNR approximation, we first recall that the use of the series representation of Bessel functions yields the following \cite{GRADSHTEYN}:
\[
x\mathrm{K}_1(x) \approx 1 +\frac{x^2}{2} \left(\ln \frac{x}{2}+c_0\right),
\]
where $c_0=-\frac{\psi(1)}{2}-\frac{\psi(2)}{2}$.
An important observation is that provided the transmission power is sufficient large,  $q(d_i, \theta_i)\rightarrow 0$. By using this observation, the integral $Q_2$ can be approximated as follows:
\begin{eqnarray}
Q_2  &=&\nonumber
e^{-(1+d_i^\alpha)\epsilon}- e^{- (1+d_i^\alpha)\epsilon} 2q(d_i, \theta_i)\mathrm{K}_1\left(2q(d_i, \theta_i)\right)  \\ \nonumber &\approx&
e^{-(1+d_i^\alpha)\epsilon}- e^{- (1+d_i^\alpha)\epsilon}\left(1+2 q^2(d_i, \theta_i)\right.\\ \nonumber &&\times\left.\left[\ln q(d_i, \theta_i) +c_0\right]\right)\\ \nonumber &=&
- e^{- (1+d_i^\alpha)\epsilon}\left(2 q^2(d_i, \theta_i) \left[\ln q(d_i, \theta_i) +c_0\right]\right)
 \\   &\approx&
- 2 q^2(d_i, \theta_i)\left[\ln q(d_i, \theta_i) +c_0\right],
\end{eqnarray}
where the last step is obtained by applying $\epsilon\rightarrow  0$. Consequently, a high-SNR approximation for the probability $Q_1$ is given by
\begin{align}
Q_1  &\approx
- \frac{2}{\pi R_{\mathcal{D}}^2}\int_{0}^{\epsilon}\int^{R_{\mathcal{D}}}_{0}\int^{2\pi}_{0}
 q^2(r, \theta)\left[\ln q(r, \theta) +c_0\right]  \\ \nonumber &\times rdrd\theta f_{x_0}(x_0) dx_0.\nonumber
\end{align}

When $R_{\mathcal{D}}<< d$, all the relay-destination distances are approximately the same as the source-destination distance, i.e.  $c_i\approx d$, $\forall i \in\{1, \ldots, N\}$.  With this approximation, we have
$$q(d_i, \theta_i)\approx \sqrt{\frac{(1+d^\alpha)(1+d_i^\alpha)(\epsilon-x_0)}{\eta}}.$$
When the transmission power becomes infinite, $\epsilon$ goes to zero, so as $\epsilon-x_0$, since $x_0\leq \epsilon$. As a result, the factor $q(d_i, \theta_i)$ goes to zero when the transmission power becomes infinite.
By using this approximation, the probability $Q_1$ in \eqref{eqpi} can now approximated as in \eqref{q1eq},
\begin{figure*}
\begin{eqnarray}
Q_1  &\approx&\label{q1eq}
- \frac{2}{\pi R_{\mathcal{D}}^2}\int_{0}^{\epsilon}\int^{R_{\mathcal{D}}}_{0}\int^{2\pi}_{0}
b_0(1+r^\alpha)\left[\frac{1}{2}\ln (b_0(1+r^\alpha)) +c_0\right]  rdrd\theta f_{x_0}(x_0) dx_0
\\ \nonumber  &=&\nonumber
- \frac{2}{  R_{\mathcal{D}}^2}\int_{0}^{\epsilon}\int^{R_{\mathcal{D}}^2}_{0}
b_0(1+z)\left[\frac{1}{2}\ln (b_0(1+z)) +c_0\right]  dzf_{x_0}(x_0) dx_0\\ \nonumber  &=&\nonumber
- \frac{2}{  R_{\mathcal{D}}^2}\int_{0}^{\epsilon}\frac{1}{b_0}\int^{b_0(1+R_{\mathcal{D}}^2)}_{b_0}
t\left[\frac{1}{2}\ln t +c_0\right]  dtf_{x_0}(x_0) dx_0,
\end{eqnarray}
\end{figure*}
where  $b_0=\frac{(1+d^2)(\epsilon-x_0)}{\eta}$  and the second equation follows from the assumption $\alpha=2$. After some algebraic manipulations, the probability $Q_1$ can be approximated as in \eqref{q12}.
\begin{figure*}
\begin{eqnarray}
Q_1  &\approx&\nonumber
- \frac{2}{  R_{\mathcal{D}}^2}\int_{0}^{\epsilon}\frac{1}{b_0}
\left[\frac{1}{4}\left(b_0^2(1+R_{\mathcal{D}}^2)^2\ln(b_0(1+R_{\mathcal{D}}^2)) - b_0^2\ln b_0 -\frac{1}{2}b_0^2R_{\mathcal{D}}^2(R_{\mathcal{D}}^2+2) \right) \right. \\ \nonumber &&\left.+\frac{c_0}{2} b_0^2R_{\mathcal{D}}^2(R_{\mathcal{D}}^2+2)\right] f_{x_0}(x_0) dx_0
\\
&\label{q12}=&
- \frac{\eta a_1}{2  R_{\mathcal{D}}^2}
\left[R_{\mathcal{D}}^2(R_{\mathcal{D}}^2+2)\left(\frac{1}{2}\frac{(1+d^2)^2\epsilon^2}{\eta^2} \ln \frac{(1+d^2)\epsilon}{\eta}-\frac{1}{2}\frac{(1+d^2)^2\epsilon^2}{\eta^2}\right)\right. \\ \nonumber &&\left.+\frac{(1+d^2)^2\epsilon^2}{2\eta^2} ((1+R_{\mathcal{D}}^2)^2\ln(1+R_{\mathcal{D}}^2)     +4e_1 R_{\mathcal{D}}^2(R_{\mathcal{D}}^2+2))\right].
\end{eqnarray}
\end{figure*}
For the special case of $\alpha=2$, the third probability in \eqref{eqpi} can be obtained by applying the following:
\begin{align}\nonumber
&\mathcal{P}(x_0<\epsilon, x_i<\epsilon) \\ \nonumber &=\frac{1}{ R_{\mathcal{D}}^2}
\int^{R^2_{\mathcal{D}}}_{0}\left(1-e^{-(1+d^2)\epsilon}\right)\left(1-e^{-(1+z)\epsilon}\right)dz
\\ \nonumber &=\frac{1}{ R_{\mathcal{D}}^2}
\left(1-e^{-(1+d^2)\epsilon}\right)\left(R_{\mathcal{D}}^2-e^{-\epsilon}\frac{1}{\epsilon} \left(1-e^{- R_{\mathcal{D}}^2 \epsilon}\right)\right)\\ \nonumber &\approx\frac{1}{ R_{\mathcal{D}}^2}
 (1+d^2)\epsilon \left(R_{\mathcal{D}}^2-(1-\epsilon)  \left( R_{\mathcal{D}}^2  -\frac{1}{2} R_{\mathcal{D}}^4 \epsilon\right)\right)\\ \label{eq44} &= \frac{1}{2}(R_{\mathcal{D}}^2+2)(1+d^2)\epsilon^2.
\end{align}
Combining \eqref{theorem}, \eqref{q12} and \eqref{eq44}, the second part of the theorem can be proved.
  \hspace{\fill}$\blacksquare$\newline

\textsl{Proof of Proposition \ref{proposition 1} :}
The starting point of the proof is to    treat the source-relay distances as constants and average out the small scale fading channels in the expression for the outage probability. Then the proof is completed by showing that the outage probability   is a decreasing function of the distance.
Particularly, because of $N\geq 1$, $ \mathcal{P}(N=0, P x_0<\tau)=0$, and the outage probability can be obtained from \eqref{eqpi}  as follows:
\begin{eqnarray}
\tilde{\mathcal{P}}&=&
\int_{0}^{\epsilon} e^{-(1+d_i^\alpha)\epsilon}\left(1-2 q(d_i, \theta_i) \mathrm{K}_1\left(2q(d_i, \theta_i)\right) \right)\\ \nonumber &&\times f_{x_0}(x_0)dx_0 +\mathcal{P}(x_0<\epsilon, x_i<\epsilon),
\end{eqnarray}
where $d_i$ and $c_i$ have been treated as constants. Following   steps similar to those used in the proof of Theorem \ref{thorem1}, the outage probability can be approximated as follows:
\begin{eqnarray}
\tilde{\mathcal{P}} &\approx&-\frac{e_2(1+d^2)\epsilon^2}{2}  \left[  \ln \epsilon-\frac{1}{2}   +  \ln e_2 +2 c_0 \right]\\ \nonumber &&
   +(1+d^2)(1+d_i^2)\epsilon^2,
\end{eqnarray}
where $e_2=\frac{(1+d^\alpha)(1+d_i^\alpha)}{\eta}$. The derivative of the outage probability in terms of  $d_i$ is given by
\begin{align}\nonumber
&\quad\quad\frac{\partial\tilde{\mathcal{P}}}{\partial d_i} \\ \nonumber&\approx
-\frac{ 2d_i (1+d^2)^2\epsilon^2}{2\eta}  \left[  \ln \epsilon-\frac{1}{2}   +  \ln \frac{(1+d^\alpha)(1+d_i^\alpha)}{\eta} +2 c_0 \right]\\ \nonumber &-\frac{e_2(1+d^2)\epsilon^2}{2}  \left[  \ln \epsilon-\frac{1}{2}   + \frac{2d_i }{ (1+d_i^\alpha)} +2 c_0 \right]
   +2(1+d^2) d_i \epsilon^2\\ \nonumber &\approx
-\frac{ 2d_i (1+d^2)^2\epsilon^2}{2\eta}    \ln \epsilon -\frac{e_2(1+d^2)\epsilon^2}{2}    \ln \epsilon<0,
\end{align}
which demonstrates that the outage probability is a decreasing function of the source-relay distance. And the proof is completed.
  \hspace{\fill}$\blacksquare$\newline

\textsl{Proof for Theorem  \ref{theorem2} :}
Conditioned on the density $\lambda_\phi$, denote $R_{i^*}$ as the relay that is closest  to the source, and $f_{d_{i^*}}(r)$  as the pdf of the corresponding shortest distance. The probability $\mathcal{P}_{d_{i^*}}(d_{i^*}>r)$ can be interpreted as the event that there is no relay located in the disc, denoted as $\mathcal{D}_{r}$, with the source at its origin and $r$ as its radius. Consequently we have
\begin{eqnarray}
\mathcal{P}_{d_{i^*}}(d_{i^*}>r) = \left.\frac{(\mu(\mathcal{D}_r))^ke^{-\mu(\mathcal{D}_r)}}{k!}\right|_{k=0}= e^{-\pi\lambda_\phi r^2 }.
\end{eqnarray}
On the other hand, the probability $\mathcal{P}_{d_{i^*}}(d_{i^*}>r)$ can be expressed as
$$\mathcal{P}_{d_{i^*}}(d_{i^*}>r)=\mathcal{P}_{d_{i^*}}(d_{i^*}>r, N\geq 1)+\mathcal{P}_{d_{i^*}}(d_{i^*}>r, N=0),$$
which leads to
\begin{align}\nonumber
&\mathcal{P}_{d_{i^*}}(d_{i^*}>r| N\geq 1) \\=&\frac{ \mathcal{P}_{d_{i^*}}(d_{i^*}>r, N\geq 1) }{\mathcal{P}_{d_{i^*}}(N\geq 1)}
\\ \nonumber =&\frac{\mathcal{P}_{d_{i^*}}(d_{i^*}>r)-\mathcal{P}_{d_{i^*}}(d_{i^*}>r, N=0)}{\mathcal{P}_{d_{i^*}}(N\geq 1)}
\\ \nonumber =&\frac{ e^{-\pi\lambda_\phi r^2 }-e^{-\pi\lambda_\phi R_{\mathcal{D}}^2 }}{1-e^{-\pi\lambda_\phi R_{\mathcal{D}}^2 }}.
\end{align}
Therefore the cumulative density function (CDF)  of $d_{i^*}$ conditioned on $N\geq1$ is $\left[1-\mathcal{P}_{d_{i^*}}(d_{i^*}>r| N\geq 1)\right]$, and the corresponding pdf is given by
\begin{eqnarray}\label{pdf zeta}
f_{d_{i^*}}(r)  &=& 2\zeta\pi\lambda_\phi r e^{-\pi \lambda_\phi r^2 },
\end{eqnarray}
where $\zeta=\frac{1}{1-e^{-\pi\lambda_\phi R_{\mathcal{D}}^2 }}$. Conditioned on $N\geq 1$, the addressed outage probability can be obtained from \eqref{eqpi} as follows:
\begin{eqnarray}\nonumber
\mathcal{P}_i &=&   \underset{Q_3}{\underbrace{\mathcal{P}\left(x_0+ \eta y_{i^*} \left( x_{i^*} - \epsilon \right) < \epsilon, x_{i^*}>\epsilon|N\geq 1\right)}} \\ \label{eqpi3} &&+\mathcal{P}(x_0<\epsilon, x_{i^*}<\epsilon| \bar{\bar{\mathcal{S}}}\geq 1).
\end{eqnarray}
Following the  steps similar to those used in the proof of Theorem \ref{thorem1}, the probability $Q_3$ can be approximated as follows:
\begin{align}
Q_3  &\approx \nonumber
\int^{R_{\mathcal{D}}}_{0}\int_{0}^{\epsilon}\left(- 2 q^2(r, \theta_i)\left[\ln q(r, \theta_i) +c_0\right]
\right)\\ \nonumber &\times f_{x_0}(x_0)dx_0 f_{d_{i^*}}(r)dr\\ \nonumber \nonumber &\approx \nonumber
- \frac{(1+d^2)\epsilon^2}{4}\int^{R_{\mathcal{D}}}_{0}
b_1  \left[2 \ln b_1\epsilon -1 +4c_0\right]  f_{d_{i^*}}(r)d r,
\end{align}
where $b_1=\frac{(1+d^2)(1+r^2)}{\eta}$. By applying the pdf $f_{d_{i^*}}(r)$ in \eqref{pdf zeta}, the probability can be approximated as in \eqref{q3eq},
\begin{figure*}
\begin{eqnarray}
Q_3 &\approx& \label{q3eq}
- \frac{(1+d^2)\epsilon^2}{4}\int^{R_{\mathcal{D}}}_{0}
b_1  \left[2 \ln b_1\epsilon -1 +4c_0\right]  2\pi\lambda_\phi \zeta r e^{-\pi \lambda_\phi r^2 }d r\\ \nonumber
&=&\nonumber
- \frac{\zeta\pi\lambda_\phi(1+d^2)^2\epsilon^2e^{\lambda_\phi \pi}}{4\eta}\int^{1+R_{\mathcal{D}}^2}_{1}
y  \left[2 \ln \frac{(1+d^2)y}{\eta}+2\ln\epsilon -1 +4c_0\right]     e^{-\pi\lambda_\phi y}d y
\\ \nonumber
&=&\nonumber
- \frac{\pi\lambda_\phi\zeta (1+d^2)^2\epsilon^2e^{\pi\lambda_\phi }}{4\eta}
 \left[b_3+2b_4\ln\epsilon -b_4(1-4c_0)\right],
\end{eqnarray}
\end{figure*}
where $b_3=2 \int^{1+R_{\mathcal{D}}^2}_{1}\ln \frac{(1+d^2)y}{\eta}e^{-\pi\lambda_\phi y}d y$ and $b_4= \frac{1}{\pi^2\lambda_\phi^2}(\gamma(2,\pi\lambda_\phi(1+R_{\mathcal{D}}^2))-\gamma(2,\pi\lambda_\phi)) $.

On the other hand, we have
\begin{align}
&\mathcal{P}(x_0<\epsilon, x_{i^*}<\epsilon|N\geq 1)\\ \nonumber \approx &    2\pi\lambda_\phi\zeta(1+d^2)\epsilon^2 \int^{R_{\mathcal{D}}}_{0} (1+r^2)r e^{-\pi\lambda_\phi r^2 }dr\\ \nonumber =&  \zeta(1+d^2)\epsilon^2 \left(1-(1+R_{\mathcal{D}}^2)e^{-\pi\lambda_\phi R_{\mathcal{D}}^2} +\frac{1}{\pi\lambda_\phi} -\frac{e^{-\pi \lambda_\phi R^2}}{\pi\lambda_\phi}\right),
\end{align}
where the approximation is obtained via $(1-e^{-x})\approx x$ for $x\rightarrow 0$.
Therefore, conditioned on $N\geq 1$, an asymptotic expression for the outage probability is given by
\begin{eqnarray}\label{eqpi2}
\mathrm{P}_i\approx - \frac{\pi\lambda_\phi\zeta (1+d^2)^2\epsilon^2e^{\pi\lambda_\phi }}{4\eta}
 \left[b_3+2b_4\ln\epsilon -b_4(1-4c_0)\right]
 \\ \nonumber + \left(1-(1+R_{\mathcal{D}}^2)e^{-\pi\lambda_\phi R_{\mathcal{D}}^2} +\frac{1}{\pi\lambda_\phi} -\frac{e^{-\pi \lambda_\phi R^2}}{\pi\lambda_\phi}\right)\\ \nonumber \times \zeta(1+d^2)\epsilon^2.
\end{eqnarray}
On noting  that  the two parameters $b_3$ and $b_4$ are not functions of $\epsilon$, the diversity gain can be obtained as follows:
\begin{eqnarray}
-\underset{P\rightarrow \infty}{\lim}\frac{\log \mathrm{P}_i}{\log P}&=& -\underset{P\rightarrow \infty}{\lim}\frac{ \log\left(\epsilon^2\ln\frac{1}{\epsilon  } \right)}{\log P}=2.
\end{eqnarray}
And the theorem is proved.
  \hspace{\fill}$\blacksquare$\newline

\textsl{Proof of Theorem \ref{theorem3} :}
The proof can be completed by first finding an upper bound of the outage probability and then showing that the achievable diversity gain is the same as the maximum diversity gain. The probability in \eqref{outage 2} can be upper bounded as follows:
\begin{eqnarray}\nonumber
Q_5 &\leq&\prod^{n}_{i=1}\mathcal{P}\left( z_{i}<\frac{\tau-x_0P}{\eta}, x_{i}>\epsilon, i\in \tilde{\mathcal{S}}\right)
\\\label{q51}
&=& \left[\mathcal{P}\left( z_{i}<\frac{\tau-x_0P}{\eta}, x_{i}>\epsilon\right)\right]^n,%
\end{eqnarray}
where $n=|\tilde{\mathcal{S}}|$ is used for notational simplicity and the condition  $i\in \tilde{\mathcal{S}}$ is removed since $x_i>\epsilon$.   Note that  the probability in \eqref{q51} is exactly the same as  \eqref{q1} when treating $x_0$ as a constant. Therefore, following similar steps in the proof in Theorem \ref{thorem1}, the probability $Q_5$ can be obtained as
\begin{eqnarray}
Q_5  \approx\nonumber
\frac{2^n}{R_{\mathcal{D}}^{2n}b_0^n}
\left[\frac{1}{4}\left(b_0^2(1+R_{\mathcal{D}}^2)^2\ln(b_0(1+R_{\mathcal{D}}^2)) - b_0^2\ln b_0\right.\right. \\ \nonumber \left.\left. -\frac{1}{2}b_0^2R_{\mathcal{D}}^2(R_{\mathcal{D}}^2+2) \right) +\frac{c_0}{2} b_0^2R_{\mathcal{D}}^2(R_{\mathcal{D}}^2+2)\right]^n.
\end{eqnarray}
Define {\small $$\mathcal{E}_1=  \underset{x_0<\frac{x_0}{P}}{ \mathcal{E}} \left\{\mathcal{P}\left( \sum^{n}_{i=1}z_{\pi_{\tilde{\mathcal{S}}}(i)}<\frac{\tau-x_0P}{\eta}, x_{\pi_{\tilde{\mathcal{S}}}(i)}>\epsilon, 1\leq i\leq n\right)\right\}.$$} By first applying some algebraic manipulations to $Q_5$,   this expectation   can be written as follows:
\begin{align}
\mathcal{E}_1
 &<
- \frac{\eta}{2^n  R_{\mathcal{D}}^{2n}}\int_{0}^{\frac{(1+d^2)\epsilon}{\eta}}
\left[\beta_1t \ln t+\beta_2 t\right]^n \left(a_1 +\eta t\right)  dt
\\ \nonumber
&=
- \frac{a_1\eta}{2^n  R_{\mathcal{D}}^{2n}}\sum^{n}_{k=0} {n \choose k}\beta_1^{k} \beta_2^{n-k} \int_{0}^{\frac{(1+d^2)\epsilon}{\eta}} t^n (\ln t)^k   dt\\ \nonumber &- \frac{\eta^2}{2^n  R_{\mathcal{D}}^{2n}}\sum^{n}_{k=0} {n \choose k}\beta_1^{k} \beta_2^{n-k} \int_{0}^{\frac{(1+d^2)\epsilon}{\eta}} t^{n+1} (\ln t)^k   dt,
\end{align}
where
$\beta_1=R_{\mathcal{D}}^2(R_{\mathcal{D}}^2+2)$ and $\beta_2=(1+R_{\mathcal{D}}^2)^2\ln(1+R_{\mathcal{D}}^2)     +4e_1 R_{\mathcal{D}}^2(R_{\mathcal{D}}^2+2)$.
Note that
\[
 \int_{0}^{\psi} t^n (\ln t)^k   dt = \psi^{n+1}\sum^{k+1}_{m=1} (-1)^{m+1} \frac{k!(\ln \psi)^{k-m+1}}{(k-m+1)!(n+1)^{m}}.
\]
By using such a result, we have
\begin{align}
\mathcal{E}_1 &\approx
- \frac{a_1\eta}{2^n  R_{\mathcal{D}}^{2n}} \beta_1^{n}    \left(\frac{(1+d^2)\epsilon}{\eta}\right)^{n+1} \frac{n!\left(\ln \frac{(1+d^2)\epsilon}{\eta}\right)^{n}}{(n+1)!}.
\end{align}
On the other hand, the probability of having $(N-n)$ relays that cannot decode the source message is given by
\begin{eqnarray}\label{outage 3}
  \mathcal{P}\left(x_{\pi_{\tilde{\mathcal{S}}^c}(j)}<\epsilon, 1\leq j\leq N-n\right) =   \left(\mathcal{P}\left(x_{i}<\epsilon\right) \right)^{N-n}.
\end{eqnarray}
Again by applying stochastic geometry, we   have
\begin{eqnarray}
\mathcal{P}\left(x_{i}<\epsilon\right) &=& \underset{\mathcal{D}}{\int} \left(1-e^{\epsilon(1+d_i^2)}\right)p_{W_i}(w_i)dw_i\\ \nonumber
&=& \frac{1}{\pi R_{\mathcal{D}}^2}\int^{R_{\mathcal{D}}}_0\int^{2\pi}_0 \left(1-e^{\epsilon(1+r^2)}\right) rdrd\theta
\\ \nonumber
&=& \frac{1}{ R^2_{\mathcal{D}}}\left(R_{\mathcal{D}}^2 - \frac{e^{-\epsilon}}{\epsilon}\left(1-e^{-\epsilon R_{\mathcal{D}}^2}\right)\right).
\end{eqnarray}
By applying the high SNR approximation, we   obtain
\begin{eqnarray}\nonumber
\mathcal{P}\left(x_{i}<\epsilon\right)
&=& \frac{1}{ R^2_{\mathcal{D}}}\left(R_{\mathcal{D}}^2 - \frac{e^{-\epsilon}}{\epsilon}\left(1-\sum_{j=1}^{\infty}(-1)^{j}\frac{\epsilon^j R_{\mathcal{D}}^{2j}}{j!}\right)\right)\\ \nonumber
&\approx& \frac{1}{ R^2_{\mathcal{D}}}\left(R_{\mathcal{D}}^2 - \frac{1-\epsilon}{\epsilon}\left(\epsilon R_{\mathcal{D}}^{2} -\frac{\epsilon^2 R_{\mathcal{D}}^{4}}{2}\right)\right)
\\
&\approx& \epsilon\left(1+\frac{R^2_{\mathcal{D}}}{2}\right).
\end{eqnarray}
Now the overall outage probability can be upper bounded as follows:
\begin{eqnarray}\label{outage 21}
 \mathcal{P_{N}} &<&   - \underset{\Pi}{\sum}  \epsilon^{N+1}\left(1+\frac{R^2_{\mathcal{D}}}{2}\right)^{N-n}
 \frac{a_1\eta}{2^n  R_{\mathcal{D}}^{2n}} \beta_1^{n}   \\ \nonumber &&\times  \left(\frac{(1+d^2) }{\eta}\right)^{n+1} \frac{n!\left(\ln \frac{(1+d^2)\epsilon}{\eta}\right)^{n}}{(n+1)!}.
\end{eqnarray}
Following   steps similar to those used in the proof of Theorem \ref{theorem2}, it can be proved that the achievable diversity gain is $(N+1)$. $(N+1)$ is the maximum diversity gain for the addressed network, which can be shown by using   a conventional cooperative network to obtain the upper bound,   and the theorem is proved.
  \hspace{\fill}$\blacksquare$\newline

\textsl{Proof of Theorem \ref{pro2} :}
By ignoring the cost, the value of $\mathcal{S}_m$ can be written as follows:
\begin{eqnarray}
v(\mathcal{S}_m) &=& \sum_{j\in \mathcal{S}_m} \frac{ SNR_{\mathcal{S}_m} -  SNR_{\mathcal{S}_m\diagup j} }{  SNR_{\mathcal{S}_m} } \\ \nonumber &=&
 \frac{  \sum_{j\in \mathcal{S}_m} \frac{P_{mj} |g_j|^2}{ {1+c_j^\alpha}} }{   \frac{P|h_{dm}|^2}{1+d_{0m}^\alpha}+\sum_{j\in \mathcal{S}_m} \frac{P_{mj} |g_j|^2}{ {1+c_j^\alpha}} }.
\end{eqnarray}
And the value of $\mathcal{S}_m$ after removing $i$ can be written as follows:
\begin{eqnarray}
v(\mathcal{S}_m\diagdown \{i\}) &=&
 \frac{  \sum_{j\in \mathcal{S}_m,j\neq i} \frac{P_{mj} |g_j|^2}{ {1+c_j^\alpha}} }{   \frac{P|h_{dm}|^2}{1+d_{0m}^\alpha}+\sum_{j\in \mathcal{S}_m,j\neq i} \frac{P_{mj} |g_j|^2}{ {1+c_j^\alpha}} }.
\end{eqnarray}
On the other hand, the value of $\mathcal{S}_n$ can be easily obtained as follows:
\begin{eqnarray}\label{step11}
v(\mathcal{S}_n) &=&
 \frac{  \frac{P_{mi} |g_i|^2}{ {1+c_i^\alpha}}}{   \frac{P|h_{dm}|^2}{1+d_{0m}^\alpha}+ \frac{P_{mi} |g_i|^2}{ {1+c_i^\alpha}} },
\end{eqnarray}
and the    value of $\mathcal{S}_n$ after removing $i$  is zero, i.e., $v(\mathcal{S}_n\diagdown \{i\}) =0$. Consequently, the overall network benefit of $\mathcal{S}_m$ and $\mathcal{S}_n$ can be obtained as follows:
\begin{eqnarray}\label{step12}
v(\mathcal{S}_m)+v(\mathcal{S}_n\diagdown  \{i\}) = \frac{  \sum_{j\in \mathcal{S}_m,j\neq i} \frac{P_{mj} |g_j|^2}{ {1+c_j^\alpha}} }{   \frac{P|h_{dm}|^2}{1+d_{0m}^\alpha}+\sum_{j\in \mathcal{S}_m} \frac{P_{mj} |g_j|^2}{ {1+c_j^\alpha}} }\\ \nonumber +\frac{    \frac{P_{mi} |g_i|^2}{ {1+c_i^\alpha}} }{   \frac{P|h_{dm}|^2}{1+d_{0m}^\alpha}+\sum_{j\in \mathcal{S}_m} \frac{P_{mj} |g_j|^2}{ {1+c_j^\alpha}} },
\end{eqnarray}
when the relay $i$ leaves $\mathcal{S}_n$ and joins in $\mathcal{S}_m$, and
\begin{eqnarray}
v(\mathcal{S}_m\diagdown \{i\})+v(\mathcal{S}_n) =
 \frac{  \sum_{j\in \mathcal{S}_m,j\neq i} \frac{P_{mj} |g_j|^2}{ {1+c_j^\alpha}} }{   \frac{P|h_{dm}|^2}{1+d_{0m}^\alpha}+\sum_{j\in \mathcal{S}_m,j\neq i} \frac{P_{mj} |g_j|^2}{ {1+c_j^\alpha}} }\\ \nonumber + \frac{  \frac{P_{mi} |g_i|^2}{ {1+c_i^\alpha}}}{   \frac{P|h_{dm}|^2}{1+d_{0m}^\alpha}+ \frac{P_{mi} |g_i|^2}{ {1+c_i^\alpha}} },
\end{eqnarray}
when the relay $i$ leaves $\mathcal{S}_m$ and joins in $\mathcal{S}_n$. Since $\phi_i(\mathcal{S}_m)< \phi_i(\mathcal{S}_n)$, we have
\begin{eqnarray}\label{step2}
 \frac{   \frac{P_{mj} |g_i|^2}{ {1+c_i^\alpha}} }{   \frac{P|h_{dm}|^2}{1+d_{0m}^\alpha}+\sum_{j\in \mathcal{S}_m} \frac{P_{mj} |g_j|^2}{ {1+c_j^\alpha}} } <   \frac{  \frac{P_{mi} |g_i|^2}{ {1+c_i^\alpha}}}{   \frac{P|h_{dm}|^2}{1+d_{0m}^\alpha}+ \frac{P_{mi} |g_i|^2}{ {1+c_i^\alpha}} }.
\end{eqnarray}
Combining \eqref{step11}, \eqref{step12} and \eqref{step2}, the proposition can be proved.
  \hspace{\fill}$\blacksquare$\newline

\end{document}